\documentclass[runningheads,envcountsame]{llncs}

\usepackage[T1]{fontenc}
\usepackage{amsmath}
\usepackage{amssymb}
\usepackage{amsfonts}

\usepackage{graphicx}
\usepackage[matrix,arrow,cmtip,curve]{xy}
\usepackage{hyperref}
\usepackage{tikz,tikz-3dplot,tikz-cd}
\usepackage{caption}
\usepackage{subcaption}
\usepackage{paralist}
\usepackage{comment}

\usetikzlibrary{fit, arrows, automata, shapes, calc, fadings,shapes.geometric}
\tikzstyle{startstop} = [rectangle, rounded corners, 
minimum width=3cm, 
minimum height=1cm,
text centered, 
draw=black, 
fill=red!30]

\tikzstyle{io} = [trapezium, 
trapezium stretches=true, 
trapezium left angle=70, 
trapezium right angle=110, 
minimum width=3cm, 
minimum height=1cm, text centered, 
draw=black, fill=blue!30]

\tikzstyle{process} = [rectangle, 
minimum width=3cm, 
minimum height=1cm, 
text centered, 
text width=3cm, 
draw=black, 
fill=orange!30]

\tikzstyle{decision} = [diamond, 
minimum width=3cm, 
minimum height=1cm, 
text centered, 
draw=black, 
fill=green!30]
\tikzstyle{arrow} = [thick,->,>=stealth]
\tikzset{->, auto, >=stealth', font=\small}
\tikzset{state/.style={shape=circle, draw, fill=white, initial text=,
    inner sep=.5mm, minimum size=1.5mm}}
\tikzset{accepting/.style=accepting by arrow}
\tikzset{state with output/.style={shape=rectangle split, rectangle
    split parts=2, draw, fill=white,
    initial text=, inner sep=1mm}}
\usepackage{tikz}
\newcommand*\loset[1]{\left[\begin{smallmatrix}#1\end{smallmatrix}\right]}
\newcommand*\bigloset[1]{\left[\begin{matrix}#1\end{matrix}\right]}

\newcommand*\intord{\dashrightarrow}

\newcommand*\ev{\textup{\textsf{ev}}}
\newcommand*{\intpom}{\mathbf{IP}}
\newcommand{\TrO}{\mathbf{T}}
\newcommand*{\Tr}{\mathsf{Tr}}
\newcommand*{\pstart}{\mathsf{start}}
\newcommand*{\pend}{\mathsf{end}}
\newcommand*{\Path}{\mathrm{Path}}
\newcommand*\subid[3]{{}_{#1}#2_{#3}}

\newcommand*\exec{%
  \raisebox{1pt}{%
    \begin{tikzpicture}[x=.8ex,y=1ex,-]
      \draw (0,0) -- (1,0) -- (1,1) -- (2,1);
    \end{tikzpicture}}}

\newcommand*\rest[1]{{}_{| #1}}
\newcommand*\pobj[1]{\square^{#1}}
\newcommand*\jneda{\mathbf{i}}
\newcommand*\jnedafinal{\mathbf{f}}
\newcommand*\ibullet{\vcenter{\hbox{\tiny $\bullet$}}}
\newcommand*\pclass{\mathbb{P}_{X}}
\newcommand*\HDA{\mathbf{HDA}}
\newcommand*\Ob{\mathbf{Ob}}
\makeatother
\newcommand\pomsetwop[4]{
  \vcenter{\xymatrix@1@R=#1@C=#2@M=#3{#4}}%
}

\newcommand\ipomset[2][1.3]{%
  \hspace*{.1em}%
  \left[%
    \hspace*{-#1em}%
    \pomsetwop{1.3ex}{#1em}{1pt}{#2}%
    \hspace*{-#1em}%
  \right]%
  \hspace*{.1em}%
}
\newcommand\bigipomset[2][1.3]{%
  \hspace*{.1em}%
  \left[%
    \hspace*{-#1em}%
    \pomsetwop{2.5ex}{#1em}{1pt}{#2}%
    \hspace*{-#1em}%
  \right]%
  \hspace*{.1em}%
}
\newcommand{\ininc}{\stackrel{0}{\hookrightarrow} }
\begin{document}
\title{Bisimulations and Logics for Higher-Dimensional Automata}
%
%
\author{Safa Zouari\inst{1}\and
Krzysztof Ziemiański\inst{2}\and
Uli Fahrenberg\inst{3}}
\authorrunning{S. Zouari et al.}
%
\institute{Norwegian University of Science and Technology, Gjøvik, Norway \email{safa.zouari@ntnu.no} \and
 University of Warsaw, Poland
\email{ziemians@mimuw.edu.pl}\\ \and
EPITA Research Lab, Rennes, France\\
\email{uli@lrde.epita.fr}}
\maketitle              
\begin{abstract}
Higher-dimensional automata (HDAs) are models of non-in\-ter\-leav\-ing
concurrency for analyzing concurrent systems. There is a rich literature
that deals with bisimulations for concurrent systems, and some of them have
been extended to HDAs. However, no logical characterizations of these
relations are currently available for HDAs.

In this work, we address this gap by introducing Ipomset modal logic, a
Hennessy-Milner type logic over HDAs, and show that it characterizes
Path-bisimulation, a variant of the standard ST-bisimulation.
We also define a notion of Cell-bisimulation, using the open-maps framework
of Joyal, Nielsen, and Winskel, and establish the relationship between
these bisimulations (and also their ``strong'' variants, which take
restrictions into account). In our work, we rely on a categorical
definition of HDAs as presheaves over concurrency lists and on track objects.
\end{abstract}


\section{Introduction}
\emph{Higher-Dimensional Automata} (HDAs), introduced by V.Pratt \cite{PrattCG} and Rob v.Glabbeek \cite{VANGLABBEEK2006265}, are a powerful model for non-interleaving concurrency. v.Glabbeek \cite{VANGLABBEEK2006265} places HDAs at the top of a hierarchy of concurrency models, demonstrating how other concurrency models, such as Petri nets \cite{nielsen1981petri}, configuration structures \cite{van1995configuration}, asynchronous systems \cite{bednarczyk1989categories}, \cite{shields1985concurrent}, and event structures \cite{winskel1986event,winskel1989introduction}, can be incorporated into HDAs. 

As its name implies, a Higher-Dimensional Automaton consists of a collection of $n$-dimensional hypercubes or $n$-cells connected via source and target maps. The well-known automata or labeled transition systems are 1-dimensional HDAs. However, HDAs allow for more expressive modeling of concurrent and distributed systems.
For example, the concurrent execution of two events $a$ and $b$ can be modeled by a square labeled as in Fig.~\ref{fig: interleaving vs true concurrency}, while an empty square represents mutual exclusion. Analogously, a filled-in 3-dimensional cube in an HDA can represent three events $a_1,a_2,$ and $a_3$ that execute concurrently, while when considering a hollow cube, each 2-dimensional face models $a_i \parallel a_j$ for $1 \leq i\neq j \leq 3$. See Fig.~\ref{fig:cube}. 

A \emph{higher-dimensional automaton} is a \emph{precubical set} together with an \emph{initial cell} and a set of \emph{final cells}. Like a simplicial set, a \emph{precubical set} is constructed by systematically gluing hypercubes. Formally, it is a graded set $X=\bigcup_{n \in \mathbf{N}}X_n$, where $X_n$ represents the set of \emph{$n$-cells}, with face maps $\delta^0$ resp $\delta^1$ 
defining the mapping of an $n$-cell to its lower resp. upper face. Each $n$-cell is associated with a linearly ordered and labeled set $V$ of length $n$. From a concurrency point of view, such a cell models a list $V$ of $n$ active events. A lower face of a cell $x_n \in X_n$ is a cell $\delta_{V \setminus U}^0(x)$ that has $U \subseteq V$ as active events. For example, in Fig.\ref{fi: face maps}, the square has active \emph{events} $[a b]$. Its faces have active events $[a]$ and $[b]$, respectively. In Section \ref{sec:HDA}, we make this precise by defining \emph{precubical sets} as presheaves over a category of linearly ordered sets with appropriate morphisms \cite{LanguageofHDA,KleeneTh}.

\tdplotsetmaincoords{70}{115}
\begin{figure}[h]
    \centering
    \begin{minipage}[t]{0.48\textwidth}
        \centering
        \begin{tikzpicture}[scale=3,tdplot_main_coords]
            \begin{scope}[shift={(1.5,0)}]
                \coordinate (O) at (0,0,0);
                \tdplotsetcoord{P}{1.414213}{54.68636}{45}
                \draw[,fill opacity=0.5] (O) -- (Py) -- (Pyz) -- (Pz) -- cycle; 
                \draw[fill=yellow!50,fill opacity=0.5] (O) -- (Px) -- (Pxz) -- (Pz) -- cycle; 
                \draw[fill=blue!50,fill opacity=0.5] (O) -- (Px) -- (Pxy) -- (Py) -- cycle; 
                \draw[fill=blue!50,fill opacity=0.5] (Pz) -- (Pyz) -- (P) -- (Pxz) -- cycle; 
                \draw[,fill opacity=0.5] (Px) -- (Pxy) -- (P) -- (Pxz) -- cycle; 
                \draw[fill=yellow!50,fill opacity=0.5] (Py) -- (Pxy) -- (P) -- (Pyz) -- cycle; 
                \node at (0.35,0,.08) {$a_1$};
                \node at (0,0.3,.06) {$a_2$};
                \node at (0,0.1,.35) {$a_3$};
                \node at (.4,.4,0) {$[a_1,a_2]$};
                \node at (.4,0,.35) {$[a_1,a_3]$};
            \end{scope}
        \end{tikzpicture}
        \caption{How HDA models concurrency: The filled-in cube models the events $[a_1,a_2,a_3]$. The 2-dimensional faces with the same color model the same 2 events. The uncolored faces model the events $[a_2,a_3]$.}
        \label{fig:cube}
    \end{minipage}
    \hfill
    \begin{minipage}[t]{0.48\textwidth}
        \centering
        \begin{tikzpicture}[x=1cm, y=1cm]
            \begin{scope}[shift={(1.5,0)}]
                \node[state] (00) at (0,0) {};
                \node[state] (10) at (-1,-1) {};
                \node[state] (01) at (1,-1) {};
                \node[state] (11) at (0,-2) {};
                \path (00) edge node[left] {$\vphantom{b}a$\,} (10);
                \path (00) edge node[right] {\,$b$} (01);
                \path (10) edge node[left] {$b$\,} (11);
                \path (01) edge node[right] {\,$\vphantom{b}a$} (11);
            \end{scope}
            \begin{scope}[shift={(5,0)}]
                \path[fill=black!15] (0,0) to (-1,-1) to (0,-2) to (1,-1);
                \node[state] (00) at (0,0) {};
                \node[state] (10) at (-1,-1) {};
                \node[state] (01) at (1,-1) {};
                \node[state] (11) at (0,-2) {};
                \path (00) edge node[left] {$\vphantom{b}a$\,} (10);
                \path (00) edge node[right] {\,$b$} (01);
                \path (10) edge node[left] {$b$\,} (11);
                \path (01) edge node[right] {\,$\vphantom{b}a$} (11);
            \end{scope}
        \end{tikzpicture}
        \caption{HDA models distinguishing interleaving $a. b + b. a$ (left) from non-interleaving concurrency $a \parallel b$ (right).}
        \label{fig: interleaving vs true concurrency}
    \end{minipage}
\end{figure}
\vspace{-0.5cm}
In addition to concurrency models, equivalence relations should also be considered when describing concurrent systems. Various notions of equivalence have been suggested in studies \cite{sangiorgi1998bisimulation,van2001refinement},\linebreak \cite{vanGlabbeek1997difference,leifer2000deriving,Van.br},\linebreak \cite{hennessy1985algebraic}, guided by considerations of the critical aspects of system behavior within a specific context and the elements from which to abstract. Parallel to behavioral equivalences, modal logic is a useful formalism for specifying and verifying properties of concurrent systems \cite{aceto2007reactive},\linebreak \cite{pnueli1992temporal,baier2008principles,baldan2020model}. 

Characterization of bisimulation in terms of Hennessy-Milner logic (HML) provides additional confidence in both approaches. Two finitely branching systems are bisimilar iff they satisfy the same logical assertions. The literature focusing on logical characterization includes the v.Glabbeek spectrum \cite{Van.br} for sequential processes and \cite{baldan2010logic,nielsen1994bisimulation},\linebreak \cite{de1995three,nielsen2005bisimulation,phillips2014event},\linebreak \cite{baldan2014hereditary} for concurrent systems. Some of these behavioral equivalences have been extended to HDAs. Among them are hereditary history-preserving bisimulations (hh-bisimulation) and ST-bisimulations \cite{VANGLABBEEK2006265}. However, to the best of our knowledge, their logical counterparts have not been investigated for HDA.

This paper presents a variant of HML interpreted over HDAs, namely \emph{Ipomset Modal Logic} (IPML). The original HML in the interleaving setting \linebreak \cite{hennessy1985algebraic} contains negation $(\neg )$, conjunction $(\wedge)$, a formula $\top$ that always holds, and a diamond modality $\langle a \rangle F$, which says that it is possible to perform an action labeled by $a$ and reach a state that satisfies $F$. Unlike the standard HML, IPML considers both sequential and concurrent computations. Thus, it differs from the standard HML within the diamond modality, so it becomes $\langle P \rangle F$ where $P$ is an \emph{interval pomset with interfaces} (interval ipomset). It is interpreted over a \emph{path} $\alpha$, and says that there is a path $\beta$ that recognizes $P$ and extends $\alpha$ to a path (concatenation of $\alpha$ and $\beta$) that satisfies $F$. For example, in Fig.~\ref{fig: interleaving vs true concurrency}, the formula $\langle (a \longrightarrow b) \rangle \top \vee \langle (b \longrightarrow a) \rangle \top$, which stands for mutual exclusion, holds at the top edge of both squares. However, formula $\langle \loset{a \\ b} \rangle \top$ holds only in the upper corner of the filled-in square (on the right). The latter formula shows that our logic is powerful enough to distinguish interleaving from true concurrency. 

\emph{Pomsets} were first introduced by Winkowski \cite{winkowski1977algebraic}. \emph{Interval orders}, a subclass of \emph{pomsets}, have been introduced by Fishburn \cite{fishburn1970intransitive}. Then, they have been equipped with \emph{interfaces} \cite{LanguageofHDA}, facilitating the definition of the \emph{gluing composition} of HDA languages.
 A computational run in an HDA is modeled by a \emph{path}, a sequence of \emph{cells}. Each two consecutive \emph{cells} are related by a \emph{source} or a \emph{target map}. The observable contents of a path $\alpha$ are described by $\ev(\alpha)$, an \emph{interval pomset with interfaces}.
\begin{figure}[h]
    \centering
    \begin{tikzpicture}[x=1.5cm, y=1.5cm]
      \begin{scope}
             \path[fill=black!15] (0,0) to (1,0) to (1,1) to (0,1);
      \node[state, initial] (00) at (0,0) {};
      \node[state] (10) at (1,0) {};
      \node[state] (01) at (0,1) {};
      \node[state] (11) at (1,1) {};
      \path (00) edge node[below] {$\vphantom{d}b$} (10);
      \path (01) edge node[above] {$b$} (11);
      \path (00) edge node[left] {$a$} (01);
      \path (10) edge node[right] {$a$} (11);
 \node at (.5,.5) {$\vphantom{b} [a\dashrightarrow b] $};
      \end{scope}
      \begin{scope}[shift={(3.5,0)}]
       \path[fill=black!15] (0,0) to (1,0) to (1,1) to (0,1);
      \node[state, initial] (00) at (0,0) {};
      \node[state] (10) at (1,0) {};
      \node[state] (01) at (0,1) {};
      \node[state] (11) at (1,1) {};
      \path (00) edge node[below] {$\vphantom{d}\delta^0_a(x)$} (10);
      \path (01) edge node[above] {$\delta^1_a(x)$} (11);
      \path (00) edge node[left] {$\delta^0_b(x)$} (01);
      \path (10) edge node[right] {$\delta^1_b(x)$} (11);
 \node at (.5,.5) {$\vphantom{b} x $};
      \end{scope}
    \end{tikzpicture}
     \caption{Two-dimensional cell with its faces. The labels of each cell are shown on the left.}
  \label{fi: face maps}
  \end{figure}
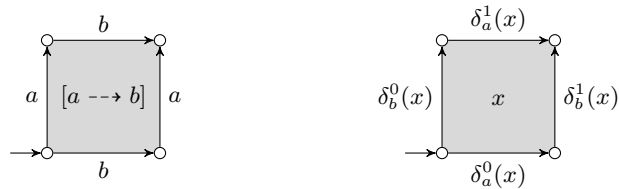
To define \emph{cell-bisimulation} and IPML over HDAs, we employ the notion of open map bisimulation \cite{JOYAL1996164}. This approach requires a category of models $\mathbf{M}$ (the category of HDAs in our case) and a path category $\TrO$ (the category of \emph{track objects} in our case), a subcategory of $\mathbf{M}$ that we call the HDA-path category. Track objects have originally been introduced in \cite{LanguageofHDA} to define languages of HDAs. They form a subcategory of $\mathbf{M}$. A \emph{track object }is a particular HDA that can be constructed from a given interval ipomset $P$, denoted $\square^P$. 
Intuitively, for a given path $\pi$ labeled with an interval ipomset $P$, the track object $\square^P$ is the smallest HDA containing $\pi$.
We show that the resulting logic characterizes path-bisimulation \cite{LanguageofHDA}, a variant of ST-bisimulation \cite{VANGLABBEEK2006265}. However, its extension, equipped with backward modality characterizes the strong path-bisimulation. Finally, we finish this paper with a hierarchy of the equivalence relations encountered in the paper. This is summarized in Fig. \ref{fi: conclusion}.
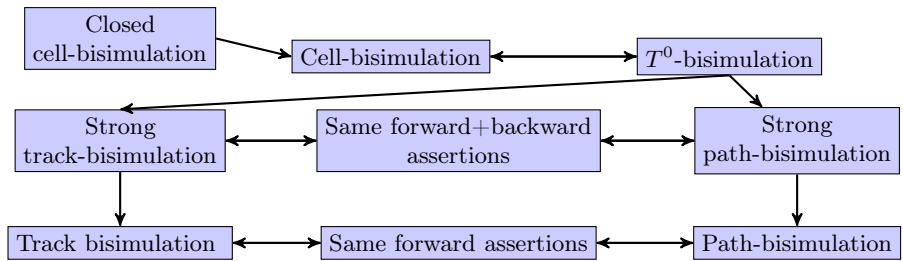
\begin{figure}[h]
\centering
\begin{tikzpicture}[fill=blue!20, x=.9cm, y=.8cm] 
\path (1.6,3.2) node(closed_cell_bis) [rectangle,rotate=0,draw,fill,align=center] {Closed \\ cell-bisimulation}
(1.6,1.5) node(strong_track) [rectangle,draw,fill,align=center] { Strong \\ track-bisimulation }
(1.6,-0.2) node(track) [rectangle,draw,fill,align=center] {Track bisimulation }
(11.6,1.5) node(strong_path_bis) [rectangle,draw,fill,align=center] {Strong \\ path-bisimulation}
(11.6,-0.2) node(path_bis) [rectangle,draw,fill,align=center] { Path-bisimulation}
(6.6,1.5) node(all_asssertions) [rectangle,rotate=0,draw,fill,align=center] {Same forward+backward \\ assertions}
(6.6,-0.2) node(asssertions) [rectangle,rotate=0,draw,fill,align=center] {Same forward assertions}
(10.6,2.9) node(t0bis) [rectangle,rotate=-0,draw,fill] {$T^0$-bisimulation}(5.6,2.9) node(cell_bis) [rectangle,rotate=-0,draw,fill] {Cell-bisimulation};
\draw[thick] (strong_track) -- (all_asssertions);
\draw[thick] (all_asssertions) -- (strong_track);
\draw[thick]  (strong_track) -- (track) ;
\draw[thick] (closed_cell_bis.east) -- (cell_bis.west) ;
\draw[thick] (t0bis) -- (cell_bis);
\draw[thick] (cell_bis) -- (t0bis);

\draw[thick] (track) -- (asssertions);
\draw[thick] (asssertions) -- (track);
\draw[thick] (t0bis.south) -- (strong_track.north);
\draw[thick] (t0bis.south) -- (strong_path_bis);
\draw[thick] (strong_path_bis) -- (path_bis);
\draw[thick] (strong_path_bis) --  (all_asssertions) ;
\draw[thick] (all_asssertions) -- (strong_path_bis);
\draw[thick] (path_bis) --  (asssertions) ;
\draw[thick] (asssertions) -- (path_bis);
\end{tikzpicture}
\caption{Hierarchy of notions of equivalence.}
\label{fi: conclusion}
\end{figure}
Other contributions, that deepen the understanding of mathematical structures and may be of independent interest, include establishing a relation between track objects and interval ipomsets in Th.\@ \ref{th: pomset and track objects}, a relation between the notion of paths of v.Glabbeek \cite{VANGLABBEEK2006265} and tracks of \cite{LanguageofHDA} in Th.\@ \ref{th: path and tracks isomorphism}, and a generalization of the Yoneda lemma to interval ipomsets in Prop.~\ref{prop: generalization of Yoneda}.
All proofs are available in a separate appendix.

\section{Higher Dimensional Automata}\label{sec:HDA}
We review the definition of Higher Dimensional Automata. We rely on a categorical approach proposed and studied in \cite{LanguageofHDA},\linebreak \cite{KleeneTh}, where an HDA is defined as a specific precubical set. To define precubical sets as presheaves over the labeled precube category $\square$, we introduce conclists and conclist maps, which are the objects and the morphisms of $\square$. We restrict our study to finitely branching HDAs.

\begin{definition}
A \emph{concurrency list} or \emph{conclist} is a tuple $(U, \dashrightarrow, \lambda)$, where $U$ is a finite set totally ordered by the strict order $\dashrightarrow$ and $\lambda: U \rightarrow \Sigma $ is a labeling map. Elements of $U$ will be called \emph{events}.
\end{definition}

\begin{definition}
\label{def: precube} 
    A \emph{conclist-map} from a conclist $U$ to $T$ is a pair $(f,\varepsilon)$
    such that:
    \begin{compactitem}
    \item 
        $f:U\to T$ is a label and order-preserving function;
    \item 
        $\varepsilon: T\rightarrow \{ 0,\exec, 1 \}$ is a function such that $\varepsilon^{-1}(\exec)= f(U)$.
    \end{compactitem}
    The composition of morphisms $(f,\varepsilon):U\to T$ and
    $(g,\zeta):T\to V$ is $(g,\zeta) \circ (f,\varepsilon)=(g\circ f,\eta)$, where
  $$  \eta(u)=
      \begin{cases}
        \varepsilon(g^{-1}(u)) & \text{for $u\in g(T)$}, \\
        \zeta(u) & \text{otherwise}.
      \end{cases} $$
     Let $\square$ be the category of conclists and conclist maps.
    We write $U\simeq V$ for isomorphic conclists;
    if two conclists are isomorphic, then the isomorphism between them is unique.
\end{definition}
For a conclist map $(f, \varepsilon): U \rightarrow V$, since the order $\dashrightarrow$ is total, $f$ is an injective function, which is determined by $V \setminus f(U)= V \setminus \varepsilon^{-1}(\exec)$, and hence by $\varepsilon$. For instance, the identity morphism $\mathbf{id}^{\square}_{V}: V \rightarrow V$ is uniquely determined by $\varepsilon_V$ where $\varepsilon_V (v)=\exec$ for all $v\in V.$ Intuitively the map $f$ injects the list of events of $U$  into the list of events of $V$, while the map $\varepsilon$ guarantees that the events of $U$ are active in $V$ by giving them the value $\exec$, and specifies the state of the remaining events by giving them the value $0$ if they are not yet started and $1$ if they are terminated. 

\paragraph*{Notation}
A morphism $(f, \varepsilon): U \rightarrow V$ might be denoted $d_{A,B}:U \to V$ where $A=\varepsilon^{-1}(0)$ and $B=\varepsilon^{-1}(1)$. Such a morphism is usually called a coface map \cite{KleeneTh}. We write $d_{A}^{0}$ for $d_{A, \emptyset}$ and $d_{B}^{1}$ for $d_{\emptyset, B}$. 
\begin{definition}
    A \emph{precubical set} $X$ is a presheaf over $\square$, 
    that is, a functor $X: \square^{\textsf{op}} \to \mathbf{Set}$.
    A \emph{precubical map} between precubical sets is a natural transformation of functors.
\end{definition}

The value of $X$ on the object $U$ of $\square$ is denoted $X[U]$. For the face map associated to coface map $d_{A, B}: U \backslash(A \cup B) \rightarrow U$, we write $\delta_{A, B}=X[d_{A, B}]: X[U] \rightarrow X[U \backslash(A \cup B)]$. Elements of $X[U]$ are cells of $X$. For any $x \in X[U],$ elements of $U$ are called events of $x$. We write $\ev(x)=U$. A precubical set is said to be finitely branching if every cell is the face of a finite number of cells.
\begin{definition}
    A \emph{higher-dimensional automaton} (HDA) $\mathcal{X}$ is a triple $(X, i_X,F_X)$ where $X$ is a precubical set, $i_X$ is a cell called the initial cell, and $F_X$ is the set of final cells. An HDA map $f: \mathcal{X} \rightarrow \mathcal{Y}$ is a precubical map $X\to Y$ that preserves the initial cell\footnote{In our study, final cells are ignored because they are not relevant for bisimulation} , that is, $f(i_X) = i_Y$. We denote $\HDA$ the category formed by HDAs as objects and HDA maps as morphisms.
\end{definition}
We assume that all HDAs are finitely branching.
  \begin{definition}\label{def: standard cube}
      Let $S$ be a conclist.
      The \emph{standard $S$-cube} is the presheaf $\square^S$ represented by $S$, that is,
\begin{compactitem}
    \item for any \emph{conclists} $T$, $\square^S[T]=\mathbf{hom}_\square(T,S)$;
    \item $\square^S [(f,\varepsilon)](g,\eta) = (g,\eta)\circ (f,\varepsilon)$
    for $(f,\varepsilon) \in \mathbf{hom}_{\square}(U,T)$,
    $(g,\eta)\in \square^S[T]$.
\end{compactitem}
 \end{definition}
 \begin{example}
    Fig.~\ref{fig:cube} and Fig.~\ref{fi: face maps} show examples of standard $S$-cubes, where $S=[a_1 \dashrightarrow a_2 \dashrightarrow a_3]$ in Fig~\ref{fig:cube} and $S=[a \dashrightarrow b]$ in Fig.~\ref{fi: face maps}.
 \end{example}
 Denote $\mathbf{y}_S \in \square^S[S]$ the cell that corresponds to the identity morphism. The following is a crucial result for this work and is implied by the Yoneda lemma.
 \begin{lemma} \label{lemma: Yoneda lemma}
     Let $X$ be a \emph{precubical set}, $S$ be a \emph{conclist}, and $x \in X$. If $\ev(x)=S$, then there exists a unique precubical map $\iota_x: \square^S \to X$ such that $\iota_x(\mathbf{y}_S)=x$.
 \end{lemma}
                                                             
 \section{Interval pomsets with interfaces vs.\ track objects}
 
 Pomsets are concurrent counterparts of words \cite{winkowski1977algebraic,PrattCG}. Interval pomsets \cite{fishburn1970intransitive} equipped with interfaces have been used to develop the language theory of HDAs \cite{MyhillNerode},\linebreak \cite{amrane2023developments,LanguageofHDA,KleeneTh}. They generalize the notion of conclist. 
Like standard cubes are presheaf representations of conclists, track objects are presheaf representations of interval pomsets with interfaces.
As we proceed in this section, we revisit these concepts, then we present important new results for constructing the HDA-path category.
\vspace{-0.35cm}
\subsubsection{Background}\label{subsec: ipomset}

A \emph{partially ordered multiset (pomset)} is a tuple $(P, <_P,\dashrightarrow_P,\lambda_P )$ where $P$ is a finite set, $\lambda_{P}: P \rightarrow \Sigma$ is a labeling function over an alphabet $\Sigma$, $<_{P}$ is a strict partial order on $P$ called \emph{precedence order}, and $\dashrightarrow_P$ is a strict partial order on $P$ called \emph{event order} such that the relation $<_P \cup \dashrightarrow_P$ is total. Elements of $P$ are called \emph{events}. Intuitively, the latter condition means that any two events in $P$ either are concurrent, thus can happen in parallel and ordered by $\dashrightarrow_P$, or occur sequentially, thus ordered by $<_P$.
For $x,y \in P$, write $x \parallel y$ if $x$ and $y$ are \emph{incomparable}, i.e, $x\neq y$, $x \nless y$, and $y \nless x$. We say that an element $x \in P$ is \emph{minimal} if there is no element $y \in P$ such that $y<x$. Similarly, we say that $x$ is \emph{maximal} if there is no element $y \in P$ such that $x<y$. 
For $Q \subseteq P$, we say that $Q$ is an \emph{antichain} if $x \parallel y$ for all $x,y \in Q$. An antichain is maximal if it is not contained in another antichain.
As the relation $<_P \cup \dashrightarrow_P$ is total, an antichain is a conclist.

A \emph{partially ordered multiset with interfaces} or \emph{ipomset} is a tuple $(P,<_{P}, \dashrightarrow_{P},$\linebreak$\lambda_{P}, S_{P}, T_{P})$, where $(P, <_P, \dashrightarrow_P,\lambda_P )$ is a pomset, $S_{P}$ is a subset of the $<$-minimal elements of $P$ called \emph{source interface}, and $T_{P}$ is a subset of the $<$-maximal elements of $P$ called \emph{target interface}. The source and target interfaces are antichains, and thus, conclists. An ipomset $P$ with empty precedence order, i.e $P=\left(U,\emptyset, \dashrightarrow_{U}, \lambda_{U}, S, T\right)$, is referred to as a \emph{discrete ipomset} and will be denoted $\subid{S}{U}{T}$. Pomsets may be regarded as ipomsets with empty interfaces. An \emph{interval ipomset} is an ipomset $P$ in which, for $x,y,z,w \in P$, if $x<z$ and $y<w$ then we have either $x<w$ or $y<z$.
\vspace{-.2cm}
    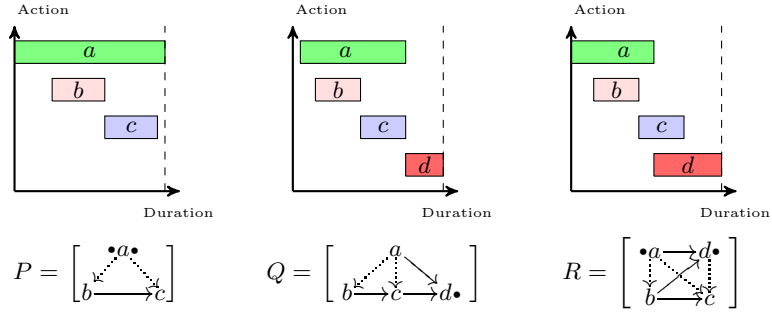
\begin{figure}
  \centering
  \begin{tikzpicture}[x=1cm]
    \def\possh{-1.3}
    \begin{scope}[shift={(0.0,0)}]
      \def\hw{0.3}
      \filldraw[fill=green!50!white,-](0,1.2)--(2,1.2)--(2,1.2+\hw)--(0,1.2+\hw);
      \filldraw[fill=pink!50!white,-](.5,0.7)--(1.2,0.7)--(1.2,0.7+\hw)--(.5,0.7+\hw)--(.5,0.7);
      \filldraw[fill=blue!20!white,-](1.2,0.2)--(1.9,0.2)--(1.9,0.2+\hw)--(1.2,0.2+\hw)--(1.2,0.2);
      \draw[thick,->](0,-.5)--(0,1.7);
      \draw[thick,->](0,-.5)--(2.2,-.5);
      \draw[thin,-,dashed](2,-.5)--(2,1.7);
       \node at (0.4,1.8) { $^{{}^{\text{Action}}}$};
       \node at (2.2,-.9) { $^{{}^{\text{Duration}}}$};
      \node at (1,1.2+\hw*0.5) {$a$};
      \node at (0.85,0.7+\hw*0.5) {$b$};
      \node at (1.55,0.2+\hw*0.5) {$c$};
      \node at (1.1,-1.8+\hw*0.7)  {$P=\bigipomset[.4]{ & \ibullet a \ibullet \ar@{.>}[dl] \ar@{.>}[dr] \\ ~ b \ar[rr] && c ~ }$}; 
    \end{scope}
    \begin{scope}[shift={(3.7,0)}]
      \def\hw{0.3}
       \filldraw[fill=green!50!white,-](0.1,1.2)--(1.5,1.2)--(1.5,1.2+\hw)--(0.1,1.2+\hw)-- (0.1,1.2);
     \filldraw[fill=pink!50!white,-](.3,0.7)--(0.9,0.7)--(0.9,0.7+\hw)--(.3,0.7+\hw)--(.3,0.7);
      \filldraw[fill=blue!20!white,-](0.9,0.2)--(1.5,0.2)--(1.5,0.2+\hw)--(0.9,0.2+\hw)--(0.9,0.2);
       \filldraw[fill=red!60!white,-](1.5,-0.3)--(2,-0.3)--(2,-0.3+\hw)--(1.5,-0.3+\hw)--(1.5,-0.3);
      \draw[thick,->](0,-.5)--(0,1.7);
      \draw[thick,->](0,-.5)--(2.2,-.5);
      \draw[thin,-,dashed](2,-.5)--(2,1.7);
       \node at (0.4,1.8) { $^{{}^{\text{Action}}}$};
       \node at (2.2,-.9) { $^{{}^{\text{Duration}}}$};
      \node at (.7,1.2+\hw*0.5) {$a$};
     \node at (0.6,0.7+\hw*0.5) {$b$};
      \node at (1.2,0.2+\hw*0.5) {$c$};
      \node at (1.75,-0.3+\hw*0.5) {$d$};
        \node at (1.1,-1.8+\hw*0.7){$Q=\bigipomset{& &  a \ar[dr] \ar@{.>}[dl] \ar@{.>}[d] &&  \\ & b \ar[r] &c \ar[r] & d \ibullet &}$}; 
    \end{scope}
    \begin{scope}[shift={(7.4,0)}]
       \def\hw{0.3}
       \filldraw[fill=green!50!white,-](0,1.2)--(1.1,1.2)--(1.1,1.2+\hw)--(0,1.2+\hw);
     \filldraw[fill=pink!50!white,-](.3,0.7)--(0.9,0.7)--(0.9,0.7+\hw)--(.3,0.7+\hw)--(.3,0.7);
      \filldraw[fill=blue!20!white,-](0.9,0.2)--(1.5,0.2)--(1.5,0.2+\hw)--(0.9,0.2+\hw)--(0.9,0.2);
       \filldraw[fill=red!60!white,-](1.1,-0.3)--(2,-0.3)--(2,-0.3+\hw)--(1.1,-0.3+\hw)--(1.1,-0.3);
      \draw[thick,->](0,-.5)--(0,1.7);
      \draw[thick,->](0,-.5)--(2.2,-.5);
      \draw[thin,-,dashed](2,-.5)--(2,1.7);
       \node at (0.4,1.8) { $^{{}^{\text{Action}}}$};
       \node at (2.2,-.9) { $^{{}^{\text{Duration}}}$};
      \node at (.7,1.2+\hw*0.5) {$a$};
     \node at (0.6,0.7+\hw*0.5) {$b$};
      \node at (1.2,0.2+\hw*0.5) {$c$};
      \node at (1.55,-0.3+\hw*0.5) {$d$};
        \node at (1.1,-1.8+\hw*0.7){$R=\bigipomset{& \ibullet a \ar[r] \ar@{.>}[d] \ar@{.>}[dr] & d \ibullet \ar@{.>}[d] & \\ & b \ar[ur] \ar[r] & c & }$}; 
    \end{scope}
  \end{tikzpicture}
  \caption{Interval ipomsets (below) with their corresponding interval representations (above). An event with a dot on the left (resp.\@ on the right) is an element of a source (resp.\@ target) interface. Full arrows indicate precedence order, while dashed arrows indicates event order.}
  \label{fi:interval rep}
\end{figure}
\begin{definition}
  Let $P$ and $Q$ be ipomsets with $T_P\simeq S_Q$. The \emph{gluing composition} of $P$ and $Q$ is $P*Q=( (P\sqcup Q)_{x\simeq f(x)}, \mathord<, \mathord{\intord},\lambda,  S_P,T_Q),$ where $(P\sqcup Q)_{x\simeq f(x)}$ is the disjoint union of $P$ and $Q$ quotiented by the unique isomorphism $f: T_P\to S_Q$, and $ \lambda( x)=
\begin{cases}
  \lambda_P( x) \text{if } x\in P,\quad ~~~~~~  &\mathord< = \mathord{<_P}\cup \mathord{<_Q}\cup( P\setminus T_P)\times( Q\setminus S_Q),\\
  \lambda_Q( x) \text{if } x\in Q,     &\mathord{\intord}=( \mathord{\intord_P}\cup \mathord{\intord_Q})^+
\end{cases}$
 

\end{definition}
\begin{definition}
    Let $P$ an ipomset. We define the order $\prec$ on maximal antichains of $P$ as follows: $U \prec T$ iff $U \neq T$ and for all $u \in U$, $t \in T$, $t \nless_P u$.
\end{definition}
\begin{proposition} (\cite{LanguageofHDA,janicki1993structure}) \label{prop: decomposition into discrete }
    Let $P$ be an ipomset. The following assertions are equivalent:
    \begin{compactenum}
        \item $P$ is an interval order;
        \item \label{eq: maximal dec} $P$ is a finite gluing of discrete ipomsets; 
         \item \label{eq: antichain decomposition} The order $\prec$ defined on the maximal antichains of $P$ is linear.
    \end{compactenum}
\end{proposition}
In this work, we focus solely on interval ipomsets: \textbf{all ipomsets are assumed to be interval even if not stated explicitly.}

 Another instrumental tool for this work is track objects. They generalize the standard cubes of Def.~\ref{def: standard cube}, replacing a conclist with an interval ipomset.
\begin{definition} \label{def:track obj}
   For a given interval ipomset $P$, the \emph{track object} is an HDA $(\square^P, i_{\square^P},f_{\square^P})$ where $\pobj{P}[U]=\hom_{\intpom}(U,P)$, $\pobj{P}[(f,\varepsilon)](g,\zeta)=(f,\varepsilon)\circ (g,\zeta)$, and $i_{ \pobj{P}}=(S_P\xrightarrow{\subseteq} P,\varepsilon)$
and $f_{ \pobj{P}}=(T_P\xrightarrow{\subseteq} P,\zeta)$
where \[ \varepsilon=
\begin{cases}
\exec & \text{if } p\in S_P, \\
0 & \text{if } p\not\in S_P,
\end{cases}
\qquad
\zeta=
\begin{cases}
\exec & \text{if } p\in T_P, \\
1 & \text{if } p\not\in T_P.
\end{cases} \]
\end{definition}
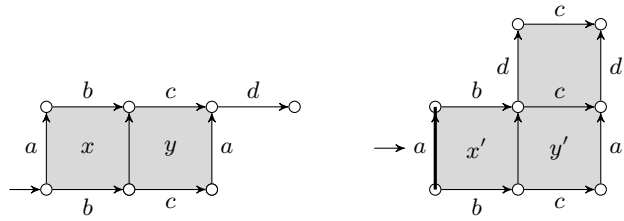
\begin{figure}
    \centering
  \begin{tikzpicture}[x=1.1cm, y=1.1cm, initial text=]
  \begin{scope}[shift={(-3.5,0)}]
        \path[fill=black!15] (0,0) to (2,0) to (2,1) to (1,1) to
      (0,1);

      \node[state] (20) at (2,0) {};
      \node[state] (10) at (1,0) {};
      \node[state] (31) at (3,1) {};
      \node[state, initial=arrow] (00) at (0,0) {};
      \node[state] (10) at (1,0) {};
      \node[state] (20) at (2,0) {};
      \node[state] (01) at (0,1) {};
      \node[state] (11) at (1,1) {};
      \node[state] (21) at (2,1) {};
      \path (00) edge node[below] {$\vphantom{d}b$} (10);
      \path (10) edge node[below] {$c$} (20);
      \path (01) edge node[above] {$b$} (11);
      \path (11) edge node[above] {$c$} (21);
      \path (00) edge node[left] {$a$} (01);
      \path (10) edge (11);
      \path (20) edge node[right] {$a$} (21);
   \path (21) edge node[above] {$d$} (31);
      \node (.5.5) at (.5,.5) {$x$};
      \node (1.5.5) at (1.5,.5) {$y$};
      \end{scope}
\begin{scope}[shift={(1.2,0)}]
        \path[fill=black!15] (0,0) to (2,0) to (2,1) to (1,1) to
      (0,1);
       \path[fill=black!15] (1,1) to (2,1) to (2,2) to (1,2) to (1,1);
       \node[state] (22) at (2,2) {};
        \node[state] (12) at (1,2) {};
      \node[state] (00) at (0,0) {};
      \node[state] (10) at (1,0) {};
      \node[state] (20) at (2,0) {};
      \node[state] (01) at (0,1) {};
      \node[state] (11) at (1,1) {};
      \node[state] (21) at (2,1) {};
      \path (00) edge node[below] {$\vphantom{d}b$} (10);
      \path (10) edge node[below] {$c$} (20);
      \path (01) edge node[above] {$b$} (11);
      \path (11) edge node[above] {$c$} (21);
      \path (00) edge node[left, initial] {$a$} (01);
      \path (10) edge (11);
      \path (20) edge node[right] {$a$} (21);
     \path (11) edge node[left] {$d$} (12);
     \path (21) edge node[right] {$d$} (22);
     \path (12) edge node[above] {$c$} (22);
      \node (.5.5) at (.5,.5) {$x'$};
      \node (1.5.5) at (1.5,.5) {$y'$};
      \draw[-, very thick, black] (0,0) -- (0,1) ;
      \end{scope}
  \end{tikzpicture}
  \caption{Example of track objects: $\square^Q$ on the left and $\square^R$ on the right, where $Q$ and $R$ are the interval ipomsets of Fig.~\ref{fi:interval rep}. The initial cell on the right is highlighted with a thick arrow.}
  \label{fig: track obj exmp}
\end{figure}
\begin{example}
Fig.~\ref{fig: track obj exmp} shows examples of track objects. For instance, the cell $x$ resp. $y$ is modeled by $(f,x) \in \pobj{Q}[U]$ resp. $(g,y) \in \pobj{Q}[T]$, where $U= [b \dashrightarrow a]$, $T= [c \dashrightarrow a]$, and \begin{gather*}
    x= \ipomset{& &  \exec \ar[dr] &&  \\ & \exec   \ar[r]  \ar@{.>}[ur] & 0  \ar@{.>}[u] \ar[r] &   0  &} ~~~~~~~~~~~~~~ \qquad%
     y= \ipomset{& &  \exec \ar[dr] &&  \\ & 1  \ar@{.>}[ur]   \ar[r] & \exec  \ar@{.>}[u] \ar[r] &   0  &}  \qquad%
\end{gather*}
It is easy to determine $f$ and $g$ as they are label preserving.
\end{example}
\subsection{The category of interval ipomsets}\label{subsec: ipomset category}
The decomposition of an interval ipomset $P$ into discrete ipomsets is not unique. However, there is a special decomposition that is unique with respect to specific properties, which we call the minimal discrete decomposition, defined as follows.
\begin{definition} \label{def: minimal discrete decomposition}
The \emph{minimal discrete decomposition} of an interval ipomset $P$ into discrete ipomsets is $P=P_1* \dots * P_m$ where $P_i=( Q_i, \emptyset, \intord\rest{P_i}, \lambda\rest{P_i}, S_i, T_i)$, $Q_i$ are the maximal antichains\footnote{In this case, we have $Q_1 \prec \dots \prec Q_m$}, $S_1= S_P$, $T_m= T_P$, and $T_i= S_{ i+ 1}= P_i\cap P_{ i+ 1}$.
\end{definition}
The minimal discrete decomposition plays a central role in our work. We rely heavily on it throughout the forthcoming proofs. Notably, such decomposition is unique.
\begin{example}\label{exm:minimal disc decom}
    The minimal discrete decompositions of interval ipomsets of Fig.~\ref{fi:interval rep} are as follows:
\begin{equation*}
  P = \bigloset{\ibullet a \ibullet \\ b} * \bigloset{\ibullet a \ibullet \\ c}; \qquad
  Q = \bigloset{a \ibullet \\ b\phantom{\ibullet}} * \bigloset{\ibullet a \\ \phantom{\ibullet}c} * d \ibullet; \qquad
  R = \bigloset{\ibullet a \ibullet \\ b} * \bigloset{\ibullet a \phantom{\ibullet}\\ \phantom{\ibullet}c \ibullet} * \bigloset{\phantom{\ibullet}d \ibullet \\ \ibullet c\phantom{\ibullet}}.
\end{equation*}

\end{example}
Since the notion of interval ipomset generalizes the concept of conclist, it is convenient to think about a category with interval ipomsets as objects that extends the $\square$ category. 
\begin{definition} \label{def: inP category}
    The category $\intpom$ consists of the following.
\begin{compactitem}
 \item Objects are interval ipomsets;
 \item A morphism\footnote{For morphisms of $\intpom$, we do not care about interfaces} between two interval ipomsets $P$ and $Q$ is a pair $(f,\varepsilon)$ such that  $f:P \to Q$ is an injective map that reflects the precedence order, that is, for $x, y\in P$, if $f( x)<_Q f( y)$ then $x<_P y$ and preserves the essential event order, i.e, for $x\parallel y\in P$ if $x \intord_P y$ then $f(x)\intord_Q f(y)$; and $\varepsilon: Q \rightarrow \{ 0,\exec,1 \}$ such that $f(P)=\varepsilon^{-1}(\exec)$ and $ \text{ if } q <_Q q' \text{ then } (\varepsilon(q),\varepsilon(q'))\in {\preceq_{ipom}},$
\end{compactitem} 
  \hspace{.5cm}  where ${\preceq_{ipom}}=\{ (1,1),(0,0),(\exec,\exec),(1,\exec),(1,0),(\exec,0)  \}$   
        \begin{compactitem}
        \item  The composition of morphisms $(f,\varepsilon):P\to Q$ and
    $(g,\zeta):Q\to R$ is \linebreak $\bigl( (g,\zeta) \circ (f,\varepsilon) \bigl)=(g\circ f,\eta)$, where $   \eta(u)=
      \begin{cases}
        \varepsilon(g^{-1}(u)) & \text{for $u\in g(Q)$}, \\
        \zeta(u) & \text{otherwise}.
      \end{cases}$
    \end{compactitem}
    We write $P\simeq Q$ if there exists a bijective map $f:P \to Q$ such that $f$ is also an order isomorphism. If such an isomorphism exists, then it is unique \cite{LanguageofHDA}.
\end{definition}
The definition of the morphisms of the category $\intpom$ will serve later to define the track objects (Def~\ref{def:track obj}). The intuition for the values of $\varepsilon(q)$ is to be $1$ if the event $q$ happens before the events of $f(P)$, $\exec$ if the event $q$ is in $f(P)$, and $0$ if the event $q$ happens after the events of $f(P)$.
That is why we allow all possible cases for $(\varepsilon(q),\varepsilon(q'))$, in $\preceq_{ipom}$, except the cases where $q'$ terminate while $q \in f(P)$ so we eliminate pairs $(\exec,1)$ and the case where $q$ has not started yet while $q'\in f(P)$ so we eliminate $(0,\exec)$.


\begin{definition}\label{def: morphism related to gluing.}
    Let $P$ and $R$ be composable ipomsets.
    There are two morphisms related to the gluing\footnote{We regard both $P$ and $Q$ as sub-pomsets of $P*Q$.} $P*R$:
    \begin{compactitem}
        \item \emph{initial inclusions} $i_P^{P*R}=(P\subseteq P*R, \varepsilon)$, and 
        \item \emph{final inclusions} $f_R^{P*R}=(R\subseteq P*R,\zeta)$, where \[
    \varepsilon(x)
    =
    \begin{cases}
        \exec & \text{for $x\in P$},\\
        0 & \text{otherwise},
      \end{cases}
      \qquad
    \zeta(x)
    =
    \begin{cases}
        \exec & \text{for $x\in R$},\\
        1 & \text{otherwise}.
    \end{cases}
  \]
\end{compactitem}
\end{definition}

\begin{definition} \label{def: inP0 category}
    For a conclist $S$, let $\intpom^0_S\subseteq \intpom$ be a subcategory with ipomsets $P$ with $S_P=S$ as objects and morphisms\footnote{ Since $i_P^{P*R}$ is uniquely determined by $R$, we might identify $i_P^{P*R}$ and $R$.}  $\hom_{\intpom^0_S}(P,Q)=\{i_P^{P*R} \mid \text{$R$ is an ipomset}$\linebreak$\text{such that $P*R\cong Q$}\}. $ 
   We define the category $\intpom^0= \bigcup_{S \in \square}\intpom^0_S$.
\end{definition}
\begin{example}\label{ex: initial inc}
The following are initial inclusions. 
\begin{compactenum}
\item $R_1: \ipomset{ & \ibullet a  \ibullet  &  \\& b &} \to \ipomset[.2]{ & \ibullet a \ibullet  \\ ~ b \ar[rr] && c~}$ where $R_1=\ipomset{ & \ibullet a \ibullet  &  \\  & c &}$;
\item $R_2: \ipomset[.2]{ & \ibullet a  \\ ~ b \ar[rr] && c ~} \to \ipomset{& \ibullet a   \ar[r] &   d \ibullet  &  \\ & b \ar[ur] \ar[r] &c & }$, where 
 $R_2=\ipomset{ & \phantom{\ibullet} d \ibullet  &  \\  & \ibullet c\phantom{\ibullet} &}$.
\end{compactenum}
\end{example}
We use categories $\intpom^0_S$ to construct the bisimulation and the modal logic later in Section~\ref{Sec: bis and ml}. In the next section, we show how we may regard $\intpom^0$ as a subcategory of $\HDA$, as required to apply the open map technique.

\subsection{Defining the HDA-path category}\label{subsec : HDA-path category}


  \begin{theorem}\label{th: pomset and track objects}
   The functor $\Tr: \intpom \rightarrow \HDA$, given by formulas $\Tr(P)=\square^P$ and
    $\Tr(f,\varepsilon)(g,\zeta)= (f,\varepsilon)\circ (g,\zeta) \text{ for } (f,\varepsilon) \in \hom_{\intpom}(P,Q)$, is faithful.
  \end{theorem}

\begin{definition}
The category $\TrO^0$ is the subcategory of $\HDA$ given by $\TrO^0 = \Tr (\intpom^0)$. Thus, it is defined by track objects as objects and morphisms are $\jneda_P^{P*R}=\Tr(i_P^{P*R})$, for $\jneda_P^{P*R}$ defined in Def.~\ref{def: morphism related to gluing.}, and called \emph{initial inclusions}.
\end{definition}
Similarly, we call $\Tr(f_P^{P*R})$, for $f_P^{P*R}$ defined in Def.\@ \ref{def: morphism related to gluing.}, a \emph{final inclusion,} and we write $\mathbf{f}_P^{P*R}$.

\begin{definition}
   Let $\square^Q$ and $\square^R$ be track objects such that $T_Q \simeq S_R \simeq U$. The \emph{gluing composition} of $\square^Q$ and $\square^R$ is the pushout HDA $(\square^Q*\square^R, I_{\square^Q},F_{\square^R})$ where  $\square^Q*\square^R=\mathrm{colim}(\square^R\stackrel{\jneda_U^{R}}{\longleftarrow} \square^U \stackrel{\jnedafinal_U^{Q}}{\longrightarrow} \square^Q ) $
\end{definition}
\begin{lemma}(\cite{LanguageofHDA}) \label{lemma: pomset pushout}
    If $Q$ and $R$ are composable ipomsets,
    then $\pobj{ Q*R} =\pobj{ Q}*\pobj{ R}$. 
In addition, $\jneda_{Q}^{Q*R}(g,\zeta)=(g,\jneda_{Q}^{Q*R}(\zeta))$ and $\jnedafinal_{R}^{Q*R}(g,\zeta)=(g,\jnedafinal_{R}^{Q*R}(\zeta))$ are given by
\begin{align*}
      \jneda_{Q}^{Q*R}(\zeta)(p)=
    \begin{cases}
      \zeta(p) & \text{for $p\in Q$}, \\
      0 & \text{otherwise},
    \end{cases} \quad \quad \quad \quad \quad &    \jnedafinal_{R}^{Q*R}(\zeta)(p)=
    \begin{cases}
      \zeta(p) & \text{for $p\in R$}, \\
      1 & \text{otherwise}.
    \end{cases}
\end{align*} 
\end{lemma}

\begin{proposition}\label{prop: finite gluing of standard cubes}
  Let $\pobj{P}$ be a track object. If $P=P_1*P_2* \dots *P_{m}$ is the minimal discrete decomposition of $P$, then $\pobj{ P}=\pobj{ P_1}* \dots * \pobj{ P_{m}}$.
\end{proposition}
Note that $\hom_{\TrO^0}(\square^P,\square^Q)\cong\{ \square^R \mid \square^Q=\square^P*\square^R \}$ by Lem.\@ \ref{lemma: pomset pushout} and Prop.\@ \ref{prop: finite gluing of standard cubes}.
    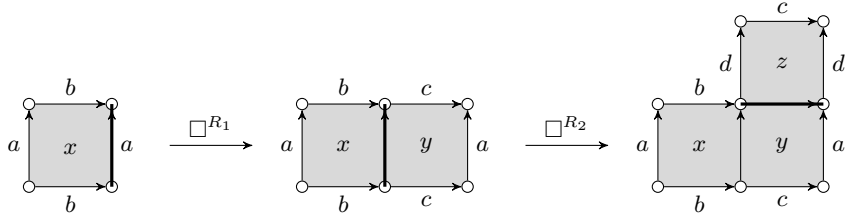
\begin{figure}[ht]
    \centering
  \begin{tikzpicture}[x=1.1cm, y=1.1cm]
  \begin{scope}[shift={(-2.9,0)}]
      \path[fill=black!15] (0,0) to (1,0) to (1,1) to (0,1);
      \foreach \x in {0, 1} \foreach \y in {0, 1} \node[state] (\x\y)
      at (\x,\y) {};
      \path (00) edge node[below] {$b$} (10);
      \path (01) edge node[above] {$b$} (11);
      \path (00) edge node[left] {$a$} (01);
      \path (10) edge node[right] {$a$} (11);
      \node at (.5,.5) {$\vphantom{b}x$};
      \path (1.7,0.5) edge node[above] {$\square^{R_1}$} (2.7,0.5);
             \draw[-, very thick, black] (1,0) -- (1,1) ;
      \end{scope}
\begin{scope}[shift={(0.4,0)}]
    \path[fill=black!15] (0,0) to (2,0) to (2,1) to (1,1) to
      (0,1);
      \node[state] (00) at (0,0) {};
      \node[state] (10) at (1,0) {};
      \node[state] (20) at (2,0) {};
      \node[state] (01) at (0,1) {};
      \node[state] (11) at (1,1) {};
      \node[state] (21) at (2,1) {};
      \path (00) edge node[below] {$\vphantom{d}b$} (10);
      \path (10) edge node[below] {$c$} (20);
      \path (01) edge node[above] {$b$} (11);
      \path (11) edge node[above] {$c$} (21);
      \path (00) edge node[left] {$a$} (01);
      \path (10) edge (11);
      \path (20) edge node[right] {$a$} (21);
       \node at (.5,.5) {$x$};
       \node at (1.5,.5) {$y$};
    \draw[-, very thick, black] (1,0) -- (1,1) ;
      \path (2.7,0.5) edge node[above] {$\square^{R_2}$} (3.7,0.5);
      \end{scope}
\begin{scope}[shift={(4.7,0)}]
    \path[fill=black!15] (0,0) to (2,0) to (2,1) to (1,1) to
      (0,1);
       \path[fill=black!15] (1,1) to (2,1) to (2,2) to (1,2) to (1,1);
       \node[state] (22) at (2,2) {};
        \node[state] (12) at (1,2) {};
      \node[state] (00) at (0,0) {};
      \node[state] (10) at (1,0) {};
      \node[state] (20) at (2,0) {};
      \node[state] (01) at (0,1) {};
      \node[state] (11) at (1,1) {};
      \node[state] (21) at (2,1) {};
      \path (00) edge node[below] {$\vphantom{d}b$} (10);
      \path (10) edge node[below] {$c$} (20);
      \path (01) edge node[above] {$b$} (11);
      \path (11) edge node[above] {} (21);
      \path (00) edge node[left] {$a$} (01);
      \path (10) edge (11);
      \path (20) edge node[right] {$a$} (21);
     \path (11) edge node[left] {$d$} (12);
     \path (21) edge node[right] {$d$} (22);
     \path (12) edge node[above] {$c$} (22);
      \node (.5.5) at (.5,.5) {$x$};
      \node (1.5.5) at (1.5,.5) {$y$};
      \node at (1.5,1.5) {$z$};
        \draw[-, very thick, black] (1,1) -- (2,1) ;
      \end{scope}
  \end{tikzpicture}
    \smallskip
  \caption{Examples of morphisms of $\TrO^0$. $R_1$ and $R_2$ are initial inclusions of Ex.\@ \ref{ex: initial inc}}
  \label{fig: track obj initial inc}
\end{figure} 
\section{Paths and tracks}
The computations or runs of HDAs, which track traversed cells and face maps, have been modeled by paths in \cite{VANGLABBEEK2006265} and by tracks in the categorical framework \cite{LanguageofHDA}. In the following subsection, we revisit these concepts and then in Subsection \ref{subsec: track and HDA-path category} establish a relation between them. This link is crucial for expressing the satisfaction relation of IPML on paths, similarly to temporal logic.


\subsection{Background: Paths and their labels}\label{subs: path and their labels}
\begin{definition}   
 A \emph{path}
 in a precubical set $X$ is a sequence
$
\alpha=(x_{0}, \varphi_{1}, x_{1}, \varphi_{2}, \ldots,$ $\varphi_{n}, x_{n}),
$
where $x_{k} \in X\left[U_{k}\right]$ are cells, and for all $k$, either
\begin{compactitem}
  \item $\varphi_{k}=d_{A}^{0} \in \square\left(U_{k-1}, U_{k}\right), A \subseteq U_{k}$ and $x_{k-1}=\delta_{A}^{0}\left(x_{k}\right)$ (up-step), or
  \item $\varphi_{k}=d_{B}^{1} \in \square\left(U_{k}, U_{k-1}\right), B \subseteq U_{k-1}, \delta_{B}^{1}\left(x_{k-1}\right)=x_{k}$ (down-step).
\end{compactitem}
\end{definition}
We write $x_{k-1} \nearrow^{A} x_{k}$ for the up-steps and $x_{k-1} \searrow_{A} x_{k}$ for the down-steps in $\alpha$. Intuitively, moving by an up step $x_{k-1} \nearrow^{A} x_{k}$  means that the list of events $A$ started and became active in the next cell $x_k$. Similarly, moving by a down step $x_{k-1} \searrow_{A} x_{k}$  means that the list of events $A$ terminated and became inactive in the cell $x_k$.
For a path written as above, we write $\pstart(\alpha)$ and $\pend(\alpha)$ for the first cell $x_0$ and the final cell $x_m$, respectively. We write $\Path_{X}$ for the set of all paths on a precubical set $X$. 

  A precubical map $f: X \rightarrow Y $ induces a map $f:\Path_{X} \to \Path_Y$. For $\alpha$ denoted as above, $f(\alpha)$ is the path $\bigl(f(x_{0}), \varphi_{1}, f(x_{1}), \varphi_{2}, \ldots, \varphi_{n}, f(x_{n})\bigl)$.
  We say that a path is \emph{sparse} if its steps are alternating between up-steps and down-steps.
   The \emph{concatenation} of $\alpha$ denoted as above and $\beta=(y_{0}, \psi_{1}, y_{1}, \psi_{2}, \ldots, \psi_{m}, y_{m})$, defined if $x_n=y_0$, is the path $\alpha*\beta$ given by $ \alpha*\beta=(x_{0}, \varphi_{1}, x_{1}, \varphi_{2}, \ldots, \varphi_{n}, x_{n}, \psi_{1},$ \linebreak $ y_{1}, \psi_{2}, \ldots, \psi_{m}, y_{m} ).$ 

\begin{definition}
     The label of a path $\alpha$ is the ipomset $\ev(\alpha)$, computed recursively:
\begin{compactenum}
  \item \label{en: one cell} If $\alpha=(x)$ has length 0, then $\ev(\alpha)={ }_{\ev(x)} \ev(x)_{\ev(x)}$.
  \item If $\alpha=\left(y \nearrow^{A}  x \right)$, where $A \subseteq \ev(x)$, then $\ev(\alpha)={}_{ \ev(x) \setminus A} \ev(x)_{\ev(x)}.$
  \item \label{en: terminator } If $\alpha=(x \searrow_B y)$, where $B \subseteq \ev(x)$, then $\ev(\alpha)={ }_{\ev(x)} \ev(x)_{\ev(x) \backslash_B}.$
  \item \label{en : label gluing} If $\alpha=\beta_{1} * \cdots * \beta_{n}$, where $\beta_{i}$ are steps, then 
 $\ev(\alpha)=\ev\left(\beta_{1}\right) * \cdots * \ev\left(\beta_{n}\right).$
\end{compactenum}
\end{definition}
As a finite gluing of discrete ipomsets, by Prop.\@ \ref{prop: decomposition into discrete }, $\ev(\alpha)$ is an interval ipomset. 
 \begin{definition}
    Let $\alpha=(x_0,\varphi_1,x_1,\dotsc,\varphi_n,x_n)$.
    We say that $\beta$ is a \emph{restriction} of $\alpha$ 
    and write $\beta \ininc \alpha$, if  $   \beta= \left(x_{0}, \varphi_{1}, x_{1}, \varphi_{2}, \ldots, x_{j-1},\varphi'_{j}, x'_{j}\right),$ where $j\leq n$ and
    \begin{compactitem}
    \item if $\varphi_j=d_B^{1}$ then $\varphi'_j=d_A^{1}$ for $A\subseteq B$;
    \item if $\varphi_j=d_B^{0}$ then $\varphi'_j=d_A^{0}$ for $A\subseteq B$
    and $x'_j=\delta^0_{B\setminus A}(x_j)$.
    \end{compactitem}
\end{definition}
\begin{example}
      On the left, the path $( \delta_a^0 x \nearrow^a x)$ in blue is a restriction of the path $( \delta_a^0 x \nearrow^a x \searrow_b \delta_b^1 x )$ in orange. On the right, the paths $( \delta_a^0\delta_b^0 x \nearrow^a \delta_b^0 x)$ and $( \delta_a^0\delta_b^0 x \nearrow^b \delta_b^0 x)$ in orange are restrictions of $\alpha_2 =( \delta_a^0 x \nearrow^a x)$ in blue.
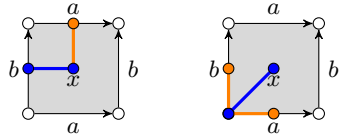
\[
   \begin{tikzpicture}[x=1.2cm, y=1.2cm]
      \path[fill=black!15] (0,0) to (1,0) to (1,1) to (0,1);
      \foreach \x in {0, 1} \foreach \y in {0, 1} \node[state] (\x\y)
      at (\x,\y) {};
      \path (00) edge node[below] {$a$} (10);
      \path (01) edge node[above] {$a$} (11);
      \path (00) edge node[left] {$b$} (01);
      \path (10) edge node[right] {$b$} (11);
      \node at (.5,.4) {$\vphantom{b}x$};
   \draw[-, very thick, blue] (0,0.5) -- (.5,0.5);
    \draw[-, very thick, orange] (0.5,0.5) -- (.5,1);
           \node[state, fill=orange, minimum size=0.2mm] at (0.5,1) {};
    \node[state, fill=blue, minimum size=0.2mm] at (0.5,0.5) {};
    \node[state, fill=blue, minimum size=0.2mm] at (0,0.5) {};
    \end{tikzpicture}
    \qquad
        \begin{tikzpicture}[x=1.2cm, y=1.2cm]
      \path[fill=black!15] (0,0) to (1,0) to (1,1) to (0,1);
      \foreach \x in {0, 1} \foreach \y in {0, 1} \node[state] (\x\y)
      at (\x,\y) {};
      \path (00) edge node[below] {$a$} (10);
      \path (01) edge node[above] {$a$} (11);
      \path (00) edge node[left] {$b$} (01);
      \path (10) edge node[right] {$b$} (11);
      \node at (.5,.4) {$\vphantom{b}x$};
      \draw[-, very thick, blue] (0,0) -- (.5,0.5);
        \draw[-, very thick, orange] (0,0) -- (.5,0);
           \node[state, fill=orange, minimum size=0.2mm] at (0.5,0) {};
                \draw[-, very thick, orange] (0,0) -- (0,0.5);
           \node[state, fill=orange, minimum size=0.2mm] at (0,0.5) {};
    \node[state, fill=blue, minimum size=0.2mm] at (0.5,0.5) {};
    \node[state, fill=blue, minimum size=0.2mm] at (0,0) {};
    \end{tikzpicture}
\]
  \end{example}
  \begin{definition}
\label{def: Congruence}
    \emph{Congruence} of paths is the equivalence relation generated by $(x \nearrow^A y \nearrow^B z) \simeq x \nearrow^{A\cup B} z$, $(x \searrow^A y \searrow^B z) \simeq x \searrow^{A\cup B} z$, and if $\alpha\simeq \alpha'$ then $\gamma*\alpha*\beta \simeq \gamma*\alpha'*\beta$.
   
    If $\alpha\simeq\beta$,
    then $\pstart(\alpha)=\pstart(\beta)$
    and $\pend(\alpha)=\pend(\beta)$.
    Furthermore, every path $\alpha$ is congruent to a unique sparse path, which is denoted $sp(\alpha)$.
\end{definition}
\begin{definition} \label{def: track}
    A \emph{track} in a precubical set $X$ 
    is a precubical map $g: \square^P \to X$ where $P$ is an ipomset. 
\end{definition}
In the case of a track in an HDA $(X,i_X,I_F)$, we say that $g$ is an initial track if $P$ is a discrete ipomset and $g(\mathbf{y}_P) = i_X$.
\vspace{-.1cm}
\subsection{The categories of tracks and paths}\label{subsec: track and HDA-path category} The relation $\simeq$ is an equivalence relation, which allows the following definition.
\begin{definition}
    Let $X$ be a precubical set. We define the category $\pclass$ as follows.
    \begin{compactitem}
        \item Objects are equivalence classes of paths with respect to $\simeq$.
        \item Morphisms are $\hom_{\pclass}(P[\alpha],[\beta])=\{[\gamma] \mid \alpha*\gamma\simeq \beta\}$), called \emph{path extensions} and denoted $e_{\alpha}^{\alpha*\gamma}=[\alpha]\to[\alpha*\gamma]$.
        \item The composition of $e_{\alpha}^{\alpha*\beta}$ and $e_{\alpha*\beta}^{\alpha*\beta*\gamma}$ is $e_{\alpha}^{\alpha*\beta*\gamma}$.
    \end{compactitem}
\end{definition}

Let $p: \pobj{P} \to X$ and $q:\pobj{Q} \to X$ be two tracks such that $p(F_P)=q(I_Q)$. 
By Lem.\@ \ref{lemma: pomset pushout}, $p$ and $q$ glue to a track $p*q: \pobj{P*Q} \to X$, called the \emph{the gluing of tracks} $p$ and $q$, that satisfies $(p*q)\circ \mathbf{i}_{P}^{P*Q}=p$ and $(p*q)\circ \mathbf{f}_{Q}^{P*Q}=q$.
\begin{definition}
    Let $X$ be a precubical set. We define the category of tracks $\mathbb{T}_X$ as follows.
    \begin{compactitem}
        \item Objects are tracks $p: \pobj{P} \to X$;
        \item 
        Morphisms are $\hom_{\mathbb{T}_X}(p,q)=\{r\mid q=p*r \}$), called \emph{track extensions} and denoted $\mathbf{e}_{P}^{P*R}$, where $p:\pobj{P} \to X$, $r:\pobj{R}\to X$ and thus $q:\pobj{P*R}\to X$. In other words, there is a morphism $\mathbf{e}_{P}^{P*R}$ between tracks $p : \square^P \to X$ and $p' : \square^Q \to X$ iff $Q=P*R$ and they are related by the left triangle in the following diagram:
\begin{tikzcd}
    \square^P\ar[r,"\jneda_P^{P*R}"]\ar[dr,"p", swap] & \square^{P*R}\ar[r,"\jneda_{P*R}^{P*R*Q}"]\ar[d,"p'"] & \square^{P*R*Q}\ar[dl,"p''"]\\
    & X &
\end{tikzcd}
        \item Composition of $\mathbf{e}_{P}^{P*R}: p \to p'$ and $\mathbf{e}_{P*R}^{P*R*Q}:p' \to p''$ is $\mathbf{e}_{P}^{P*R*Q}:p \to p''$, as shown in the diagram above.
        
    \end{compactitem}
\end{definition}
The Yoneda lemma \ref{lemma: Yoneda lemma} is based on the unique cell $\mathbf{y}_S$ of a conclist $S$. For an ipomset $P$, we introduce the characteristic path $\rho_P$ that allows a generalization of the Yoneda lemma by substituting $\rho_P$ for $\mathbf{y}_S$. It is a key contribution of this work that will be used to bridges tracks and paths.
\begin{definition}\label{def: char path}
    Consider an interval ipomset $P$ and $P=P_1*P_2* \dots *P_{m}$ its minimal discrete decomposition with $P_i=\subid{U\setminus A}{U}{U\setminus B}$. The characteristic path of $P$ is $\rho_P= \beta_1*\beta_2 \dots *\beta_m$ the concatenation of steps $\beta_i=(\delta^0_A(\mathbf{y}_U),\mathbf{y}_U,\delta^1_B(\mathbf{y}_U))$.

    Note that the characteristic path is the sparse path $\rho_P \in P_{\pobj{P}}$ such that $\ev(\rho)=P$, $\pstart(\rho)=I_{\square^P}$, and $\pend(\rho)=F_{\square^P}$, calculated by induction as follows.
        \begin{compactitem}
        \item  If $P=\subid{U\setminus A}{U}{U\setminus B}$ is discrete, then $\rho_P=(\delta^0_A(\mathbf{y}_U),\mathbf{y}_U,\delta^1_B(\mathbf{y}_U))$.
        \item If $P=R*Q$, then $\rho_P=\jneda_R^{R*Q}(\rho_R)*\jnedafinal_Q^{R*Q}(\rho_Q)$\footnote{We can check that $\jneda_R^{R*Q}(\rho_R)$ and $\jnedafinal_Q^{R*Q}(\rho_Q)$ can be concatenated by elementary calculations, using the expression of initial and final inclusion of Lem.\ \ref{lemma: pomset pushout} and of the initial and final cells in Def.\ \ref{def:track obj}.}, where $\rho_R \in P_{\square^R}$ and $\rho_Q \in P_{\square^Q}$ the characteristic paths of $R$ and $Q$ respectively. 
        \end{compactitem} 
\end{definition}

The following Prop.\@ generalizes Yoneda Lemma \ref{lemma: Yoneda lemma}. Instead of cells, here we have paths, and instead of conclist, we have ipomsets.
\begin{proposition} \label{prop: generalization of Yoneda}
    Let $X$ be a precubical set, $P$ an ipomset, $ \alpha \in \Path_{X}$. If $\ev (\alpha)=P$, then there exist a unique $\rho_P'\simeq \rho_P$ and a unique track $g_{\alpha} :\square^P \rightarrow X$ such that $g_{\alpha} (\rho_P')=\alpha$.
\end{proposition}
    The track $g_{\alpha}$ depends on the class of $\alpha$ up to $\simeq$ rather than $\alpha$:
\begin{lemma}\label{lemma: equiv clas Yoneda}
    If $\alpha \simeq \beta$, then $g_{\alpha}=g_{\beta}$.
\end{lemma}

\begin{theorem}\label{th: path and tracks isomorphism}
    For any precubical set $X$,
    the categories $\pclass$ and $\mathbb{T}_X$ are isomorphic.
\end{theorem}
\vspace{-.5 cm}
\section{Bisimulation and modal logic}\label{Sec: bis and ml}
      Fix a conclist $S$ and denote by $\HDA_S$ the full subcategory of $\HDA$ with HDAs $(X,i_X,F_X)$ such that $\ev(i_X)=S$.
      So that we have $\HDA=\bigcup_{S\in \square}\HDA_S$.     Similarly, $\TrO_S^0$ is the category that has track objects $\pobj{P}$, with $S_P\simeq S$, as objects and initial inclusions as morphisms. In this section, we apply the open map bisimulation technique \cite{JOYAL1996164} with $\TrO^0_S$ as the HDA-path category to define the \emph{$\TrO^0$-bisimulation} and then to define the IPML. 
  \subsubsection{Overview} A morphism $\varphi: \mathcal{X} \to \mathcal{Y}$ in $\HDA_S$ has the \emph{path lifting-property} with respect to $\TrO_S^0$ if whenever for $\mathbf{i}_{P}^{Q} \in \hom_{\TrO^0_S}$ (thus $Q\cong P*R$ for some ipomset $R$), $p: \pobj{P} \rightarrow X$ and $q: \pobj{Q} \rightarrow Y$, $q \circ \mathbf{i}_{P}^{P*R} = \varphi \circ p$ i.e the following diagram commutes, \[
    \xymatrix{%
      \pobj{P} \ar[r]^{p}
      \ar[d]_{\mathbf{i}_{P}^{P*R}} \ar@{}[dr]  &
      X \ar[d]^{\varphi} \\
      \pobj{P*R}  \ar[r]_{q}  \ar@{.>}[ur]_{p'} & Y  
    }
\]
then there exists a track $p'$ such that $p'\circ  \mathbf{i}_{P}^{P*R}=p$ and $\varphi \circ p'=q$ i.e the two triangles in the previous diagram commute. In this case, we say that $\varphi$ is  $\TrO_S^0$\emph{-open} or that $\varphi$ is \emph{open with respect to} $\TrO_S^0$. 
This gives rise to a notion of bisimulation with respect to $\TrO_S^0$. 
\subsection{Bisimulation from open maps for HDA}
 \begin{definition} \label{def: U bis from open map}
Let $\mathcal{Y}$, $\mathcal{Z}$ be HDAs. We say that $\mathcal{Y}$ and $\mathcal{Z}$ are $\TrO_S^0$-bisimilar if there is a span of  $\TrO_S^0$-open HDA maps $\mathcal{Y} \stackrel{\varphi}{\longleftarrow} \mathcal{X} \stackrel{\psi}{\longrightarrow} \mathcal{Z}$ with a common HDA $\mathcal{X}$.
     \end{definition} 
A path $\alpha$ in an HDA $\mathcal{X}$ is a path in the precubical set $X$ such that $\pstart(\alpha)=i_X$. We denote $\Path_{\mathcal{X}}$ the set of paths in $\mathcal{X}$. Similarly, we denote $\mathbb{P}_{\mathcal{X}}$ the category of classes of paths in an HDA.
A morphism $\varphi: \mathcal{X} \to \mathcal{Y}$ in $\HDA$ has the \emph{future path lifting property} if for $\alpha \in \Path_{\mathcal{X}}$ and $\beta \in \Path_Y$, if $\varphi(\alpha) \text{ and } \beta$ can be concatenated, then there exists $\alpha'$ in $X$ such that $\alpha \text{ and } \alpha^{\prime}$ can be concatenated and $\varphi(\alpha * \alpha^{\prime})=\varphi(\alpha)* \beta$.

\begin{lemma} \label{Lemma: Zig Zag property for U}
For any HDA map $\varphi: \mathcal{X} \rightarrow \mathcal{Y}$, $\varphi$ is $\TrO^0_S$-open iff $\varphi$ has the future path lifting. 
\end{lemma}

\begin{definition}\label{def: U bar bis}
     A \emph{closed cell-bisimulation} between HDAs $\mathcal{Y}$ and $\mathcal{Z}$ is a relation $\overline{R}$ between cells in $Y$ and $Z$ such that
\begin{compactenum}
\item \label{en: U bar bisim.e1} initial cells $i_Y$ and $i_Z$ are related;
\item \label{en: U bar bisim.e3} $\overline{R}$ respects labels: for all
  $( y, z)\in \overline{R}$, $\ev_Y( y)= \ev_Z(z)$;
\item \label{en: U bar bisim.e4} if $( y, z)\in \overline{R}$, then $( \delta^{\nu}_A(y),
 \delta^{\nu}_A(z))\in \overline{R}$ for $A \subseteq \ev_Y(y)= \ev_Z(z)$, $\nu \in \{0,1\}$;
\item \label{en: U bar bisim.e5} for all $( y, z)\in \overline{R}$, if there exists $y'$ such that $\delta^{0}_A(y')=y$ for some $A\subseteq \ev(y')$, then there exists $z'$ such that $\delta^{0}_A(z')=z$ and $( y', z') \in \overline{R}$; 

\item \label{en: U bar bisim.e6} for all $( y, z)\in \overline{R}$, if there exists $z'$ such that $\delta^{0}_A(z')=z$ for some $A\subseteq \ev(z')$, then there exists $y'$ such that $\delta^{0}_A(y')=y$ and $( y', z') \in \overline{R}$; 
\end{compactenum}
\end{definition}

A cell $x$ in an HDA $\mathcal{X}$ is said to be accessible if there exists $\alpha_x \in \Path_{\mathcal{X}}$ such that $\pend(\alpha_x)=x$, we denote $\mathcal{X}_{acc}$ the set of accessible cells in $\mathcal{X}$.

\begin{definition}\label{def: Cell bis}
     A \emph{cell-bisimulation} between $\mathcal{Y}$ and $\mathcal{Z}$ is a relation $R$ between $\mathcal{Y}_{acc}$ and $\mathcal{Z}_{acc}$ that satisfies the same properties as Def.\@ \ref{def: U bar bis}, replacing \ref{en: U bar bisim.e4}. by
\begin{compactenum}
\setcounter{enumi}2
\item for all $( y, z)\in R$, for all $A \subseteq \ev_Y(y)= \ev_Z(z)$,
\begin{compactenum}
    \item \label{en: Cell bisim.e4} $( \delta^{1}_A(y),
 \delta^{1}_A(z))\in R$
    \item  $\delta^{0}_A(y) \in \mathcal{Y}_{acc}$ iff $,
 \delta^{0}_A(z)\in \mathcal{Z}_{acc}$. In this case, $( \delta^{0}_A(y),
 \delta^{0}_A(z))\in R$
\end{compactenum}
\end{compactenum}
\end{definition}

\begin{theorem} \label{U and T bis}
Two HDAs $\mathcal{Y}$ and $\mathcal{Z}$
    are $\TrO^0_S$-bisimilar iff they are cell-bisimilar.  
\end{theorem}

\subsection{Modal characterization}
In this section, we delve into the core contributions of our work. Initially, we introduce the concept of track bisimulation, followed by a formal presentation of the Ipomset modal logic. Notably, the notion of track bisimulation serves as a crucial link connecting our logic’s modalities with the existing concept of ST-bisimulation found in the literature. More specifically, it will demonstrate that our logic characterizes the notion of ST-bisimulation.

\def\R{\mathsf{R}}
\begin{definition}\label{def: Track-bisimulation}
     A \emph{track-bisimulation}, with respect to $\TrO^0_S$, between HDAs $\mathcal{Y}$ and $\mathcal{Z}$ is a symmetric relation $\R$ of pairs of tracks $\left(p_{1}, p_{2}\right)$ with common domain $\square^P$, so $p_{1}: \square^P \rightarrow Y$ is a track in $Y$ and $p_{2}: \square^P \rightarrow Z$ is a track in $Z$, such that
\begin{compactenum}
    \item \label{eq1: track bis} initial tracks $\iota_{\mathcal{X}}:\square^S\to X$
    and $\iota_{\mathcal{Y}}:\square^S\to Y$ are related;
    \item \label{eq2: track bis} For $\left(p_{1}, p_{2}\right) \in \R$, if $p_{1}^{\prime} \circ \jneda_P^{P*R}=p_{1}$, in the diagram
 
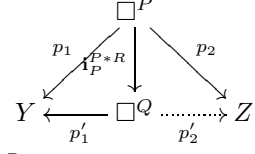
\xymatrix{
     &    \square^P \ar[dl]_{p_1} \ar[d]_{\jneda_P^{P*R}} \ar[dr]^{p_2}   \\
   Y & \square^Q \ar[l]^{p'_1 } \ar@{.>}[r]_{p'_2} &  Z                          }

then there is $p'_2$ such that $(p'_1,p'_2) \in \R$ and $p'_2 \circ\jneda_P^{P*R}=p_2$. 
\end{compactenum}
We say that a track-bisimulation is \emph{strong} if, in addition, it satisfies:
\begin{compactenum}
\setcounter{enumi}2
\item \label{eq3: track bis}  If $(p_1,p_2) \in \R$ for $p_{1}: \square^Q \rightarrow Y$, $p_{2}: \square^Q \rightarrow Z$,
then for every $\jneda_P^{P*R}: \square^P \rightarrow \square^Q\in \TrO^0_S$
we have $(p_1 \circ \jneda_P^{P*R},p_2 \circ \jneda_P^{P*R}) \in \R$.
\end{compactenum}
We say that two HDAs are (strong) \emph{track-bisimilar} iff there is a \emph{(strong) track-bisimulation} between them.
\end{definition}
We introduce the novel modal logic IPML with HDAs as models, following the approach of Nielsen and Winskel \cite{JOYAL1996164}.  
\begin{definition}[Ipomset Modal Logic]
    The set of formulae in Ipomset Modal Logic (IPML) is given by the following syntax:
\[
F,G::= \top \mid \bot \mid F \wedge G \mid F \vee G  \mid\langle \jneda_P^{P*R} \rangle F\mid  \overline{\langle \jneda_P^{P*R}\rangle} F,
\]
where $\jneda_P^{P*R}$ is a morphism in $\TrO^0_S$. The modality $\overline{\langle \jneda_P^{P*R}\rangle}$ is a backward modality, while $\langle \jneda_P^{P*R}\rangle$ is a forward modality. 
\end{definition}
Like Nielsen and Winskel's original approach, IPML should also have infinitary conjunctions. In contrast, we only consider HDAs with finitely branching, for which no infinitary conjunction is required.

The satisfaction relation between a track $p: \square^P \rightarrow X$ and a formula $F$ is given by structural induction on assertions as follows:
 \begin{compactitem}
\item $p \models \top$ for all $p$, $p \models \bot$ for no $p$, $p \models F \wedge G$ iff $p \models F $ and $p \models G$, and $p \models F \vee G$ iff $p \models F $ or $p \models G$;
\item $p \models \langle \jneda_P^{P*R} \rangle F$ where $\jneda_P^{P*R}: \square^{P} \rightarrow \square^{P*R}$, iff there exists is a track $q: \square^{P*R}  \rightarrow X$ for which $q \models F$ and $p=q \circ \jneda_P^{P*R}$;
\item $p \models \overline{\langle \jneda_Q^{P}\rangle} F$ where $\jneda_Q^{P}: \square^Q \rightarrow \square^{P} $, with $P=Q*S$, iff there exists a track $q: \square^{Q}  \rightarrow X$ for which
$q \models F$ and $q= p \circ \jneda_Q^{Q*S}$.
\end{compactitem}
By Th. \@ \ref{th: path and tracks isomorphism}, the previous modal logic, given with satisfaction relation on tracks, induces a modal logic interpreted over paths, where congruent paths satisfy the same formulas. The induced satisfaction relation is thus a binary relation $\models$ that relates $\alpha \in \mathbb{P}_{\mathcal{X}}$, with $\ev(\alpha)=P$, to formulae. 

For a given track, the forward modality is uniquely determined by the choice of the extending ipomset $R$. While the backward modality is uniquely determined by the decomposition of $P$ into two ipomsets. Thus, our modalities could be reformulated, as follows.
 \begin{compactitem}
 \item $\alpha \models \langle R \rangle F$ with $R$ an ipomset iff there is $\beta \in \pclass$ for which $\alpha*\beta \models F$ and $\ev(\beta)=R$;
\item $\alpha \models \overline{\langle Q*S \rangle} F$ with $P=Q*S$ iff there is $\alpha' \ininc \alpha$ in $\mathbb{P}_{\mathcal{X}}$ for which $\ev(\alpha')=Q$ and $\alpha' \models F$.
\end{compactitem}

\begin{example}
   Consider the paths in the HDAs of Fig.\ref{fig: counterexample }. We have the following:
\begin{compactitem}
    \item $(i_{X_2}) \models  \bigl\langle \loset{\phantom{\ibullet} c \ibullet \\ \ibullet a \ibullet} \bigr\rangle \bigl\langle  [ \ibullet a \to d \ibullet ] \bigr\rangle \top$, meaning that there is a path $\alpha_2$ labeled by $\loset{\phantom{\ibullet} c \ibullet \\ \ibullet a \ibullet}$ from which it is possible to terminate an event labeled by $a$ and start an event labeled by $d$, by executing the path $\beta_2$.
    \item $\alpha_1 \models  \overline{ \bigl\langle c * \loset{d \\ a} \bigr\rangle} \bigl\langle \loset{b \\ d} \bigr\rangle \top$, meaning that there exists a restriction $\alpha'_1$ of $\alpha_1$ such that $\ev(\alpha'_1)=c$ and $\alpha'_1 \models \bigl\langle \loset{b \\ d} \bigr\rangle \top$, that is, $\alpha_1'$ can be concatenated with a path $\beta_1$ labeled by $\loset{b \\ d}$.
\end{compactitem}
\end{example}

\cite[Thm.~15]{JOYAL1996164} now immediately implies the following.

\begin{theorem}
HDAs 
are (strong) track-bisimilar iff initial tracks satisfy  
the same forward (and backward) assertions.
\end{theorem}

\begin{theorem}\label{th: cell to track}
   If HDAs 
   are $\TrO^0_S$-bisimilar, then they are strong track-bisimilar.
 \end{theorem}
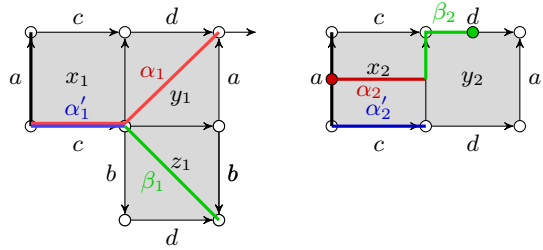
\begin{figure}[h]
    \centering
  \begin{tikzpicture}[x=1.25cm, y=1.25cm]
  \begin{scope}[shift={(-2,0)}]
        \path[fill=black!15] (0,0) to (2,0) to (2,1) to (1,1) to
      (0,1);
        \path[fill=black!15] (2,0) to (2,-1) to (1,-1) to
      (1,0);
      \node[state] (20) at (2,0) {};
      \node[state] (2-1) at (2,-1) {};
      \path (20) edge node[right] {$b$} (2-1);
      \node[state] (1-1) at (1,-1) {};
      \node[state] (10) at (1,0) {};
      \path (10) edge node[left] {$b$} (1-1);
      \path (20) edge node[right] {$b$} (2-1);
      \node[state] (2-1) at (2,-1) {};
      \node[state] (1-1) at (1,-1) {};
      \path (1-1) edge node[below] {$d$} (2-1);
      \node[state] (00) at (0,0) {};
      \node[state] (10) at (1,0) {};
      \node[state] (20) at (2,0) {};
      \node[state] (01) at (0,1) {};
      \node[state] (11) at (1,1) {};
      \node[state, accepting] (21) at (2,1) {};
      \path (00) edge node[below] {$\vphantom{d}c$} (10);
      \path (10) edge node[below] {} (20);
      \path (01) edge node[above] {$c$} (11);
      \path (11) edge node[above] {$d$} (21);
      \path (00) edge node[left] {$a$} (01);
      \path (10) edge (11);
      \path (20) edge node[right] {$a$} (21);
      \node (.5.5) at (.5,.5) {$x_1$};
      \node (1.6.3) at (1.6,.3) {$y_1$};
      \node (1.6-.4) at (1.6,-.4) {$z_1$};  
      \draw[-, very thick, black] (0,0) -- (0,1) ;
       \draw[-, very thick, red!75] (0,0.03) -- (1,0.03) --  (2,1);
         \draw[-, very thick, blue!75] (0,0) -- (1,0);
         \draw[-, very thick, green!80!black] (1,0) --  (2,-1);
           \node at (1.3,.6) {$\vphantom{b} \textcolor{red!80!black}{\alpha_1}$};
           \node at (1.3,-.6) {$\vphantom{b} \textcolor{green!80!black}{\beta_1}$};
            \node at (.5,.2) {$\vphantom{b} \textcolor{blue!80!black}{\alpha_1'}$};
      \end{scope}
\begin{scope}[shift={(1.2,0)}]
        \path[fill=black!15] (0,0) to (2,0) to (2,1) to (1,1) to
      (0,1);
      \node[state] (00) at (0,0) {};
      \node[state] (10) at (1,0) {};
      \node[state] (20) at (2,0) {};
      \node[state] (01) at (0,1) {};
      \node[state] (11) at (1,1) {};
      \node[state] (21) at (2,1) {};
      \path (00) edge node[below] {$\vphantom{d}c$} (10);
      \path (10) edge node[below] {$d$} (20);
      \path (01) edge node[above] {$c$} (11);
      \path (11) edge node[above] {$d$} (21);
      \path (00) edge node[left] {$a$} (01);
      \path (10) edge (11);
      \path (20) edge node[right] {$a$} (21);
      \node (.5.6) at (.5,.6) {$x_2$};
      \node (1.5.5) at (1.5,.5) {$y_2$};
      \draw[-, very thick, red!80!black] (0,0.5) -- (.5,.5) --  (1,.5);
       \draw[-, very thick, blue!80!black] (0,0) -- (1,0) ;
     \draw[-, very thick, green!80!black] (1,0.5) -- (1,1) -- (1.5,1);
            \node at (.4,.4) {$\vphantom{b} \textcolor{red!80!black}{\alpha_2}$};
             \node at (1.2,1.2) {$\vphantom{b} \textcolor{green!80!black}{\beta_2}$};
      \draw[-, very thick, black] (0,0) -- (0,1) ;
       \node[state, fill=red!80!black, minimum size=1.5mm] at (0,0.5) {};
       \node[state, fill=green!80!black, minimum size=1.5mm] at (1.5,1) {};
       \node at (.5,.2) {$\vphantom{b} \textcolor{blue!80!black}{\alpha_2'}$};
      \end{scope}%
  \end{tikzpicture}
   \smallskip
  \caption{Two HDAs $\mathcal{X}_1$ and $\mathcal{X}_2$ in $\HDA_{a}$ that are strong track-bisimilar, cell-bisimilar, but not closed cell-bisimilar HDAs. Each HDA has the edge labeled by $a$ with a thick line as the unique initial cell.}
  \label{fig: counterexample }
\end{figure}
\vspace{-.4cm}
   It is clear that closed cell-bisimilarity implies cell-bisimilarity. The following example shows that the opposite direction is false. It also shows that track-bisimilarity does not imply closed cell bisimilarity.  
The opposite direction of the later remains an open question that we would like to answer in an extended version of this work.
\begin{example}
 Fig.~\ref{fig: counterexample } shows two HDAs $\bigl(X_1,\delta_c^0(x_1)\bigr)$ (on the left) and $\bigl(X_2,\delta_c^0(x_2)\bigl)$ (on the right) that are track-bisimilar but not closed cell-bisimular.  Let $P_1=\loset{\ibullet a \ibullet \\ c}$, $P_2=\loset{\ibullet a\phantom{\ibullet} \\ \phantom{\ibullet}d \ibullet}$, $P_3=\loset{\ibullet d \\ \phantom{\ibullet}b}$.
      It is not difficult to check that $K=\{(g,g')\mid g:\pobj{P} \to X_1,g':\pobj{P} \to X_2 \mid \text{there exists } i_{P}^{P_1*P_2 } \}$ is a strong track-bisimulation. However, they cannot be closed cell-bisimilar, because if there is a closed cell-bisimulation between them, then $\delta^0_a(y_1)$ and $\delta^0_a(y_2)$ are related. However, $\delta_b^0(z_1)=\delta^0_a(y_1)$ while there exists no cell $z_2 \in X_2$ such that $\delta_b^0(z_2)=\delta^0_a(y_2)$.
\end{example}

\begin{remark}
    To check the track bisimilarity of the HDAs of the previous example, one may check that initial \emph{paths} satisfy the same forward assertions. Note that if we allow $i_{X_1}$ and $i_{X_2}$ to be the nodes $\delta_{ac}^0(x_1)$ and $\delta_{ac}^0(x_2)$ respectively, $\mathcal{X}_1$ and $\mathcal{X}_2$ will no longer be track bisimilar. Due to the distinguishing formulae $ \bigl\langle  (c) \bigr\rangle \bigl\langle  \left[\begin{smallmatrix}b\\d \end{smallmatrix}\right] \bigr\rangle \top$ that holds in $(i_{X_1})$ but not in $(i_{X_2})$. In fact, in this case, we will have equivalence between the notions of strong track bisimilarity and cell-bisimilarity. 
\end{remark}

     We say that paths $\alpha=(x_{0}, \varphi_{1}, x_{1}, \varphi_{2}, \ldots, \varphi_{n}, x_{n})$ and $\beta=(y_{0}, \psi_{1}, y_{1}, \psi_{2}, \ldots,$\linebreak $ \psi_{m}, y_{m})$ have the same shape if $n=m$ and $\varphi_i =  \psi_i$ for all $i$. The following notion of behavioral equivalence was originally introduced by v.Glabbeek \cite{VANGLABBEEK2006265} as ST-bisimulation. In our setting it has been formulated in \cite{LanguageofHDA} as follows.
 
 \begin{definition} \label{def: U bisimulation}
 A \emph{path-bisimulation} between HDAs $\mathcal{Y}$ and $\mathcal{Z}$ is a symmetric relation $\R$ between paths in $Y$ and $Z$ such that
 \begin{compactenum}
\item \label{en: U bisim.e1}initial paths $(i_Y)$ and $(i_Z)$ are related;
\item \label{en: U bisim.e3} $\R$ respects the shape: for all
  $( \rho, \sigma)\in \R$, $\rho$ and $ \sigma$ have the same shape;
\item \label{en: U bisim.e4} for all $( \rho, \sigma)\in \R$ and path
  $\rho'$ in $Y$ where $\rho$ and $\rho'$ may be concatenated,
  there exists a path $\sigma'$ in $Z$ such that $( \rho * \rho',
  \sigma* \sigma')\in \R$;
 \end{compactenum}
  A path-bisimulation is called \emph{strong} if, in addition, it satisfies:
\begin{compactenum}
\setcounter{enumi}3 
\item \label{en: U bisim.e6} for all $( \rho, \sigma)\in \R$ and $\rho'$ a restriction of $\rho$, there exists $\sigma'$ a restriction of $\sigma$ such that $( \rho',\sigma')\in \R$.
\end{compactenum}
Finally, $\mathcal{X}$ and $\mathcal{Y}$ are (strong) \emph{\emph{path-bisimilar}} if there exists a (strong) \emph{path-bisimulation} $\R$ between them; this is an equivalence relation.
\end{definition}

\begin{theorem}\label{th: last theorem}
   Two HDAs $\mathcal{X}_1$ and $\mathcal{X}_2$ are (strong) track-bisimilar iff they are (strong) path-bisimilar.
\end{theorem}
 
\paragraph*{Conclusion}
We have investigated open maps for the category $\TrO^0\subseteq \TrO$. The general approach yields the abstract notion of $\TrO^0$-bisimulation and a path logic, IPML (with past modality) for which the logical equivalence is equivalent to the (strong) Track-bisimulation. We showed that our logic is powerful enough to capture true concurrency and characterize (strong) Path-bisimulation, known in the literature as ST-bisimulation. We summarize the hierarchy of the different notions in Fig.\ref{fi: conclusion}. In future work, we aim to look at the extension of IPML that captures the finest bisimulation equivalence, hereditary history preserving bisimulation. We would thus have a complete spectrum for concurrency bisimulation notions that might be interpreted over other models of concurrency such as Petri nets, event structures, and configuration structures, due to the expressiveness of HDA.
\bibliographystyle{ACM-Reference-Format}

\newpage
\appendix
   \begin{figure*}[bp]
\begin{tikzpicture}[fill=blue!20, x=.9cm] 
\path  (2.8,3.2) node(closed_cell_bis) [rectangle,rotate=0,draw,fill] {Closed cell-bisimulation} (7.5,3.2) node(cell_bis) [rectangle,rotate=0,draw,fill] {Cell-bisimulation}
(1.6,1.5) node(strong_track) [rectangle,draw,fill] { strong track bisimulation }
(1.6,-0.2) node(track) [rectangle,draw,fill] {track bisimulation }
(13.6,1.5) node(strong_path_bis) [rectangle,draw,fill] { Strong Path-bisimulation}
(13.6,-0.2) node(path_bis) [rectangle,draw,fill] {Path-bisimulation}
(7.6,1.5) node(all_asssertions) [rectangle,rotate=0,draw,fill] {Same forward+backward assertions}
(7.6,-0.2) node(asssertions) [rectangle,rotate=0,draw,fill] {Same forward assertions}
(11.6,3.2) node(t0bis) [rectangle,rotate=-0,draw,fill] {$T^0$-bisimulation};
\draw[thick] (strong_track) -- (all_asssertions);
\draw[thick] (all_asssertions) -- (strong_track);
\draw[thick]  (strong_track) -- (track) ;
\draw[thick] (closed_cell_bis) -- (cell_bis);
\draw[thick] (cell_bis) -- (closed_cell_bis) ;
\draw[thick] (t0bis) -- (cell_bis);
\draw[thick] (cell_bis) -- (t0bis) ;
\draw[thick] (closed_cell_bis) -- (strong_track);
\draw[thick] (strong_track) -- (closed_cell_bis);

\draw[thick] (track) -- (asssertions);
\draw[thick] (asssertions) -- (track);
\draw[thick] (t0bis.south) -- (strong_path_bis);
\draw[thick] (strong_path_bis) -- (path_bis);
\draw[thick] (strong_path_bis) -- (t0bis.south);
\draw[thick] (strong_path_bis) --  (all_asssertions) ;
\draw[thick] (all_asssertions) -- (strong_path_bis);
\draw[thick] (path_bis) --  (asssertions) ;
\draw[thick] (asssertions) -- (path_bis);
\end{tikzpicture}
\caption{Hierarchy of notions of equivalence in the particular case $S=\emptyset$.}
\label{fi: conclusion2}
\end{figure*}
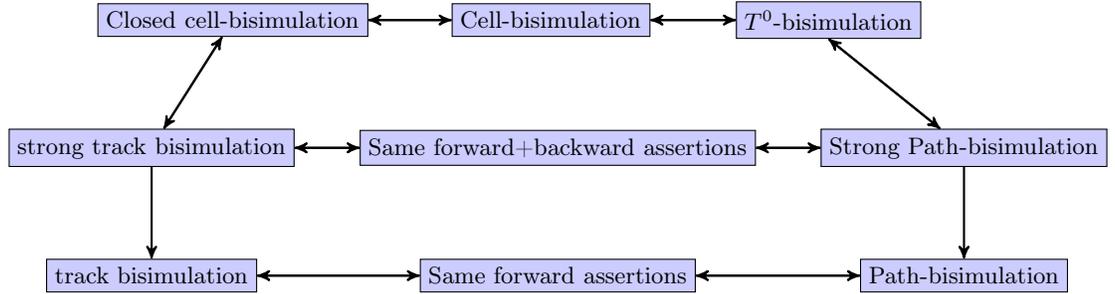

   \section{Particular case}\label{Sec: App particular case}

    In this section, we consider $S$ as the empty conclist, meaning that we allow initial cell to be a node.
    \begin{lemma}
        Let $\alpha=\left(x_{0}, \varphi_{1}, x_{1}, \varphi_{2}, \ldots, \varphi_{n}, x_{n}\right) \in \Path_{\mathcal{X}}$ where $A\subseteq \ev(x_n)$. There exists $\alpha'\in \Path_{\mathcal{X}}$ such that $\pend(\alpha')=\delta^0_A(x_m)$. 
    \end{lemma}

\begin{proof}
    We will show the case where $A=\{a\}$. The general case could be shown by iterating following the same principle and using the relation $\delta_{A\cup B}^{0}=\delta_{A}^{0}\delta_{B}^{0}$. Since $|\ev(x_0)|=0$ and there exists $k$ such that $\varphi_k=\square (U_{k-1},U_k)$, $\delta_a(x_k)=x_{k-1}$, and $a \in \ev(x_k) $ for all $i \geq k$, so that  $\alpha=\left(x_{0}, \varphi_{1}, x_{1}, \varphi_{2}, \ldots,\varphi_{k-1},\delta_a (x_{k}), \varphi_{k},x_k, \ldots, \varphi_{n}, x_{n}\right)$. 
    Let $\alpha'=\left(x_{0}, \varphi_{1}, x_{1}, \varphi_{2}, \ldots,\varphi_{k-1},\delta_a (x_{k}),\delta_a^0(x_k), \ldots, \varphi'_{n}, \delta_a^0(x_{n})\right)$ 
    where $\ev(x_{k})=U_{k}$ and for all $i\geq k+1$, either
\begin{compactitem}
  \item $\varphi'_{i}=d_{A}^{0} \in \square\left(U_{i-1}\setminus {a}, U_{i}\setminus {a}\right), A \subseteq U_{i} \setminus {a}$ and $\delta_a^0(x_{i-1})=\delta_{A}^{0}\left(\delta_{a}^{0}(x_{i})\right)$ (up-step), or
  \item $\varphi'_{i}=d_{B}^{1} \in \square\left(U_{i}, U_{i-1}\right), B \subseteq U_{i-1}, \delta_{B}^{1}\left(\delta_{a}^{0}(x_{i-1)}\right)=\delta_{a}^{0}(x_{i})$ (down-step).
\end{compactitem}

\end{proof} 
We will denote by $\delta_A^0(\alpha)$ the path constructed in the proof.

\begin{theorem}
    Two HDAs $\mathcal{Y}$ and $\mathcal{Z}$ are Cell-bisimilar iff they are strong Path-bisimmilar.
\end{theorem}

\begin{proof}
    $\Rightarrow$ This is equivalent to Th. \@ \ref{th: cell to track} and can be shown in a similar way.

    $\Leftarrow$ Assume that there is a Path-bisimulation $R$ between $\mathcal{Y}$ and $\mathcal{Z}$. We can easily check that the relation $\overline{R}=\bigl\{(\pend( \delta_A^0(\alpha)),\pend(\delta_A^0(\beta))\mid (\alpha,\beta)\in R $ $ \text{and } A\subseteq \ev(\pend(\alpha)) \bigl\}$
    is a Cell-bisimulation
\end{proof}
    The hierarchy in this case is shown in Fig.~\ref{fi: conclusion2}.
\section{Omitted Proofs}\label{omitted proofs}
  \begin{proof}[Proof of Thm.\@ \ref{th: pomset and track objects}]
Let $(f_1,\varepsilon_1),(f_2,\varepsilon_2): P \to Q$ be morphisms of $\intpom$ such that $\Tr(f_1,\varepsilon_1)=\Tr(f_2,\varepsilon_2): \square^P \to \square^Q$. Using the expression of the composition in Def.\@ \ref{def: inP category}, we have $(f_1 \circ g, \eta_1)=(f_2 \circ g, \eta_2)$ for all $(g,\zeta)$. In particular for $(g,\zeta)=\mathbf{id}_{\intpom}^P$, i.e, $g:P \to P$ such that $g(p)=p$ for all $p$ and $\zeta^{-1}(\exec)=P$, 
we obtain $f_1=f_2$ and $\eta_1=\eta_2$, that is \begin{equation*}
   \begin{cases} \zeta(f_1^{-1}(q)) & \text { if } q \in f_1(P), \\ \varepsilon_1(q) & \text { otherwise}. \end{cases}= \begin{cases} \zeta(f_2^{-1}(q)) & \text { if } q \in f_2(P), \\ \varepsilon_2(q) & \text { otherwise}. \end{cases}
    \end{equation*}
On the other hand, since 
$
  \varepsilon_i(q)= \begin{cases} \exec & \text { if } q \in f_i(P), \\ \varepsilon_i(q) & \text { otherwise}. \end{cases}  $, $\varepsilon_1=\varepsilon_2$. Thus $(f_1,\varepsilon_1)=(f_2,\varepsilon_2)$.
\end{proof}

\begin{proof}[Proof of Prop.\@ \ref{prop: finite gluing of standard cubes}]
   We employ induction on $m$ and use Lem.\@ \ref{lemma: pomset pushout}.
\end{proof}

\begin{proof}[Proof of Prop.\@ \ref{prop: generalization of Yoneda}]
Induction on the number of cells of $\alpha$.
    \begin{compactitem}
        \item If $\alpha=(x)$, then $\rho_P'=(\mathbf{y}_P)$, and the Yoneda embedding $\iota_{x}$ satisfy the requirements.
        \item If $\alpha=(y\nearrow^A x)$, then $P=\subid{\ev(x) \setminus A}{\ev(x)}{\ev(x)}$ and $\rho_P'=(\delta^0_A(\mathbf{y}_U),\mathbf{y}_U)$, where $U=\ev(x)$. Again, we take $g_{\alpha}=\iota_x$.
         \item If $\alpha=(x\searrow^B y)$, then $P= \subid{\ev(x)}{\ev(x)}{\ev(y)\setminus B} $ and $\rho_P'=(\mathbf{y}_U,\delta^1_A(\mathbf{y}_U))$, where $U=\ev(x)$. Similarly, we consider $g_{\alpha}=\iota_x$.  
        \item If $\alpha=\beta*\theta$ where both $\beta$ and $\theta$ are shorter than $\alpha$. Let $R=\ev(\beta)$ and $Q=\ev(\theta)$ (we therefore have $P=R*Q$). By induction hypothesis, there exist $\rho_R'\simeq\rho_R$, $\rho_Q'\simeq \rho_Q$, a unique $g_{\beta}:\square^R \rightarrow X$, and a unique $g_{\theta}:\square^Q \rightarrow X$, such that $g_{\beta}(\rho_R')=\beta$ and $g_{\theta}(\rho_Q')=\theta$. $g_{\beta}*g_{\theta}$ satisfies the requirement of $g_{\alpha}$. By Def.\@ \ref{def: char path}, $\rho_R'*\rho_Q'\simeq \rho_P$.
    \end{compactitem}
    \end{proof} 
    \begin{proof}[Proof of Lem.\@ \ref{lemma: equiv clas Yoneda}]
Induction on the length of $\alpha$.
    \begin{compactitem}
        \item If $\alpha=(x\nearrow^{A\cup B} z)$ and $\beta=(x\nearrow^{A} y \nearrow^{ B} z)$ , then both $g_{\alpha}$ and $g_{\beta}$ are the Yoneda embedding $\iota_z$.
        \item Similarly, if $\alpha=(x\searrow_{A\cup B} z)$ and $\beta=(x\searrow_{A} y \searrow_{ B} z)$.  
        \item If $\theta*\alpha*\gamma \simeq \theta*\beta*\gamma $ such that $\alpha\simeq \beta$. We have $g_{\theta*\alpha*\gamma}=g_{\theta}*g_{\alpha}*g_{\gamma}$ and $g_{\theta*\beta*\gamma}=g_{\theta}*g_{\beta}*g_{\gamma}$.
        By induction hypothesis, $g_{\alpha}=g_{\beta}$, thus $g_{\theta*\alpha*\gamma}=g_{\theta*\beta*\gamma}$.\qed
    \end{compactitem}
\end{proof}
\begin{proof}[Proof of Thm.\@ \ref{th: path and tracks isomorphism}]
    Let us define a functor $\Psi: \pclass \to  \mathbb{T}_X$ by
    $\Psi([\alpha])= g_{\alpha}$, $\Psi(e_\alpha^{\alpha*\beta})=\mathbf{e}_{\ev(\alpha)}^{\ev(\alpha)*\ev(\beta)}$. By Lem.\@ \ref{lemma: Upsilon is a functor} in App. \ref{omitted proofs}, $\Upsilon$ is a functor. To show that it is the inverse of $\Psi$, let $\alpha \in \pclass$, $\Upsilon \circ \Psi(\alpha)= \Upsilon (g_\alpha)=g_\alpha(\rho_P)=\alpha$ by Lem.\@ \ref{lemma: equiv clas Yoneda} and Prop.\@ \ref{prop: generalization of Yoneda}. We use the same argument to show that $\Upsilon \circ \Psi(e_\alpha^{\alpha*\beta})=e_\alpha^{\alpha*\beta}$.
To show that $\Psi \circ \Upsilon = \mathbf{id}_{\mathbb{T}_X}$, we have $\Psi \circ \Upsilon(p)=\Psi (p(\rho_P))=g_{p(\rho_P)}=p$ by the uniqueness of such a track (Lem.\@ \ref{lemma: equiv clas Yoneda}). The same argument applies to show that $\Psi \circ \Upsilon(\mathbf{e}_{P}^{P*R})=(\mathbf{e}_{P}^{P*R})$. 
\end{proof}
\begin{lemma}\label{lemma: varphi commutes}
    For any HDA map $\varphi: \mathcal{X} \rightarrow \mathcal{Y}$ and $\alpha \in \Path_{\mathcal{X}}$, we have $\varphi \circ g_{\alpha}=g_{\varphi(\alpha)}$.
\end{lemma}
\begin{proof}[Proof of Lem.\@ \ref{lemma: varphi commutes}]
    Let $P=\ev(\alpha)$ and $\rho_P'$ the path of Prop.\@ \ref{prop: generalization of Yoneda}. By definition, $g_{\alpha}(\rho_P')=\alpha$, thus $\varphi \circ g_{\alpha} (\rho_P')=\varphi(\alpha)$. We obtain the equality required by the uniqueness of $g_{\varphi(\alpha)}$.
\end{proof}
\begin{proof}[Proof of Lem.\@ \ref{Lemma: Zig Zag property for U}]
$1 \Rightarrow 2$ Let $\ev (\alpha)=P$ and $\ev(\beta)=R$. By the definition of $g_{\varphi(\alpha)*\beta}$ (as constructed int the proof of Prop.\@ \ref{prop: generalization of Yoneda}), we have $g_{\varphi(\alpha)*\beta}\circ \mathbf{i}_{P}^{P*R} = g_{\varphi(\alpha)}=\varphi \circ g_{\alpha}$ (by Lem.\@ \ref{lemma: varphi commutes}), meaning that the following diagram commutes
\[    \xymatrix{%
      \square^{P} \ar[r]^{g_{\alpha}}
      \ar[d]_{\mathbf{i}_{P}^{P*R}} \ar@{}[dr]  &
      X \ar[d]^{\varphi} \\
      \square^{P*R}  \ar[r]_{g_{\varphi(\alpha)*\beta}} \ar@{.>}[ur]_{g'}  & Y  
    }
\]
Since $\varphi$ is $\TrO^0$-open, there exists $g': \square^{P*R} \to Y$ such that $g' \circ \mathbf{i}_{P}^{P*R}= g_{\alpha}$ and $\varphi \circ g' = g_{\varphi(\alpha)*\beta}$.
Consider $\alpha'=g'(\mathbf{f}_{Q}^{P*Q}(\rho_{R}'))$.
We have $g' \circ \mathbf{i}_{P}^{P*R}(\rho_P')=\alpha$, thus $\alpha$ and $\alpha'$ can be concatenated. On the other hand, $\varphi(\alpha')=(g_{\varphi(\alpha)}*g_{\beta}) \circ \mathbf{f}_{Q}^{P*Q}(\rho'_R)=g_{\beta}(\rho'_R)=\beta$, by the definition of the track concatenation.

$2\Rightarrow 1$ Let $\varphi: \mathcal{X} \rightarrow \mathcal{Y}$ be an HDA map, and $q: \square^{P*R} \rightarrow Y$ such that $\varphi \circ p = q \circ \jneda_P^{P*R}$. Let $\rho_P$ and $\rho_Q$ be the characteristic paths of $P$ and $Q$ respectively. By definition, $\rho_Q= \jneda_P^{P*R}(\rho_P)* \sigma$ for some $\sigma \in \Path_{\pobj{Q}}$, hence 
$q(\rho_Q)=\bigl(q\circ \jneda_P^{P*R} (\rho_P)\bigl)* q(\sigma)$. Defining $\alpha=p(\rho_P)$ and $\beta=q(\sigma)$, we obtain $q(\rho_Q)=\varphi(\alpha)* \beta$. The \emph{future path lifting property} yields a \emph{path }$\alpha'$ such that $\varphi(\alpha*\alpha')=\varphi(\alpha)*\beta$. Since $\ev(\alpha*\alpha')= \ev(q(\rho_Q))=Q$, by Prop.\@ \ref{prop: generalization of Yoneda}, there exists a unique track 
$g: \pobj{Q} \to X$ such that $g(\rho_Q)=\alpha*\alpha'$. On one hand $g \circ \jneda_P^{P*R} (\rho_P)= \alpha = p(\rho_P)$, 
thus by Prop.\@ \ref{prop: generalization of Yoneda}, $g \circ \jneda_P^{P*R} = p$. On the other hand, 
$\varphi \circ g (\rho_Q)= \varphi (\alpha*\alpha') = q (\rho_Q)$ thus again by Prop.\@ \ref{prop: generalization of Yoneda} $\varphi \circ g = q$. Therefore, $\varphi$ is $\TrO^0_S$-open.  
\end{proof}
\begin{proof}[Proof of Thm.\@ \ref{U and T bis}]
    $\Rightarrow$: Assume that there is a span of $\TrO^0_S$-open HDA-maps $\mathcal{Y} \stackrel{\varphi}{\longleftarrow} \mathcal{X} \stackrel{\psi}{\longrightarrow} \mathcal{Z}$. 
    The relation $K=\{ (\varphi(x), \psi(x)) \mid x \in \mathcal{X}_{acc} \}$ is a Cell-bisimulation. Since $\varphi(i_X)=i_Y$ and $\psi(i_X)= i_Z$, \ref{def: Cell bis}. \ref{en: U bar bisim.e1} is satisfied. By \cite[Lemma 27]{LanguageofHDA}, $K$ respects labels, thus \ref{def: Cell bis}.\ref{en: U bar bisim.e3} is satisfied. 
    Condition \ref{def: Cell bis}.\ref{en: U bar bisim.e4} is satisfied as $\varphi$ is a precubical map. To show \ref{def: Cell bis}.\ref{en: U bar bisim.e5} and similarly \ref{def: Cell bis}.\ref{en: U bar bisim.e6}, let $(y,z)=(\varphi(x),\psi(x))$. Assume that there exists $y' \in Y$ such that $\delta_A^0(y')=y$, meaning that $\varphi(\alpha_x)$ can be concatenated with $\beta=(y \nearrow^A y')$.    Since $\varphi$ is open, by Lem.\@ \ref{Lemma: Zig Zag property for U}, there exists $\alpha_x'=(x \nearrow^A x')$ in $X$ such that $\varphi(\alpha')=(y \nearrow^A y')$. Defining $z'=\psi(x')$ we obtain $\delta_A^0(z')=z$ and $(y',z') \in K$.

$\Leftarrow$: Assume that there exists a Cell-bisimulation $\R$ between $\mathcal{Y}$ and $\mathcal{Z}$. Let $\mathcal{X}=(X,(i_Y, i_Z))$ where $X=R$ and $\delta_A^{\nu}(y,z)=(\delta_A^{\nu}(y),\delta_A^{\nu}(z))$. Let $\varphi$ and $\psi$ be projections which for $(y,z) \in X$ give $y$ and $z$ respectively. By condition \ref{def: Cell bis}.\ref{en: U bar bisim.e1}, $\varphi$ and $\psi$ preserve initial cells.  Let $(y,z)\in X$, and then an up-step $\beta=\bigl(\varphi(y,z) \nearrow^A y'\bigl)$ in $Y$. 
 As $(y,z) \in \R$, by \ref{def: Cell bis} .\ref{en: U bar bisim.e5}, there exists $z' \in Z$ such that $z \nearrow^A z'$ and $(y',z') \in R$. That is, there exists an up-step $\alpha=((y,z) \nearrow^A (y',z'))$ in $X$ such that $\varphi(\alpha)=\beta$. Thus, $\varphi$ has the future path-lifting property. By Lem.\@ \ref{Lemma: Zig Zag property for U}, $\varphi$ is  $\TrO^0_S$-open. We can show that $\psi$ is $\TrO^0_S$-open similarly.
\end{proof}
\begin{proof}[Proof of Thm.\@ \ref{th: cell to track}]
 Assume that there is a span of  $\TrO^0_S$-open maps with a common HDA $\mathcal{X}$. We show that the relation $K=\{ (\varphi \circ p, \psi \circ p); p\text{ is a track in }X \}$ is a strong Track-bisimulation. Since $\varphi$ and $\psi$ preserve initial cells, $K$ satisfies \ref{def: Track-bisimulation}.\ref{eq1: track bis}. To show that $K$ satisfies the condition \ref{def: Track-bisimulation}.\ref{eq2: track bis}, assume that $ (p_1,p_2) \in K$, so that $ (p_1,p_2)=(\varphi \circ p, \psi \circ p)$ for some track $p$ in $X$. If $p_1= p_1' \circ \jneda_P^{P*R}$ for some morphism $\jneda_P^{P*R}: \square^P \to \square^Q$ of $\TrO^0_S$, that is, the left square in the following diagram commutes
\[ \xymatrix{
     & ~~ \square^P \ar[d]_{p} \ar[dl]_{\jneda_P^{P*R}} \\ \square^Q \ar@{.>}[r]^{p'} \ar[d]_{p'_1}  &  X \ar[dl]_{\varphi} \ar[dr]^{\psi}   \\  Y &  &  Z  } \]
Then, since $\varphi$ is $\TrO^0_S$-open, there exists $p': \square^Q \to X$ such that the two triangles in the previous diagram commute. Let $p'_2= \psi \circ p'$, we have $(p'_1,p'_2) \in K$ and $p'_2 \circ \jneda_P^{P*R}=\psi \circ p' \circ \jneda_P^{P*R}= \psi \circ p=p_2$ as required by \ref{def: Track-bisimulation}.\ref{eq2: track bis}. Finally, if $(p_1,p_2) \in K$, it is clear that for any initial inclusion $\jneda_P^{P*R}: \square^P \to \square^Q$ $ (p_1 \circ \jneda_P^{P*R}, p_2 \circ \jneda_P^{P*R}) \in K$. Thus, the Track-bisimulation $K$ is strong. 
\end{proof}

\begin{lemma}\label{lemma: Upsilon is a functor}
    The map $\Upsilon:\mathbb{T}_X \to  \pclass$ given as follows:
    \begin{compactitem}
        \item $\Upsilon(p)= p(\rho_P)$ for $p: \pobj{P}\to X \in \Ob(\mathbb{T}_X)$;
        \item $\Upsilon(f)= [p(\rho_P)] \to \bigl[p(\rho_P)*\bigl(p'\circ \jnedafinal_{R}^{P*R}(\rho_R)\bigl)\bigl]$ for $f: p \to p'$, where $p: \pobj{P} \to X$ and $p': \pobj{P*R} \to X$
    \end{compactitem}
    is a functor.
\end{lemma}

\begin{proof}
To show that $\Upsilon$ is a functor, 
on the one hand, we have
\begin{align}\label{proofeq1}
\Upsilon(\mathbf{e}_{P}^{P*R*Q})(p)
&= [p(\rho_P)] \to [p(\rho_P)*p''\circ \jnedafinal_{R*Q}^{P*R*Q}(\rho_{R*Q})]
\\
&=[p(\rho_P)*p''\circ \jnedafinal_{R*Q}^{P*R*Q}(\jneda_{R}^{R*Q} (\rho_R)* \jnedafinal_{Q}^{R*Q}(\rho_Q))] \\
& =\bigl[p(\rho_P)*\bigl(p''\circ \jnedafinal_{R*Q}^{P*R*Q}\circ \jneda_{R}^{R*Q} (\rho_R)\bigl) * \bigl( p''\circ \jnedafinal_{R*Q}^{P*R*Q}\circ \jnedafinal_{Q}^{R*Q}(\rho_Q)\bigl)\bigl].
\end{align}
     On the other hand, we have\footnote{Note that $\Upsilon(f)(p)=p'(\rho_{P*R})$ because $\rho_{P*R}=\jneda_{P}^{P*R} (\rho_P)* \jnedafinal_{R}^{P*R}(\rho_R)$ and $ p'\circ \jneda_{P}^{P*R}=p$.} 
 \begin{align*}
\Upsilon(\mathbf{e}_{P}^{P*R})=& [p(\rho_P)] \to \bigl[p(\rho_P)*\bigl(p'\circ \jnedafinal_{R}^{P*R}(\rho_R)\bigl)\bigl]=[p'(\rho_{P*R})]\\
\Upsilon(\mathbf{e}_{P*R}^{L})=& \bigl[p'(\rho_{P*R})] \to \bigl[(p(\rho_{P})*p'\circ \jnedafinal_{R}^{P*R}(\rho_R))*p''\circ \jnedafinal_{Q}^{L}(\rho_Q)\bigl].
 \end{align*}
 where $L=P*R*Q$. Thus, 
 \begin{equation}\label{proofeq2}
    \bigl(\Upsilon(\mathbf{e}_{P*R}^{P*R*Q})\circ \Upsilon(\mathbf{e}_{P}^{P*R})\bigl)(p)= (p(\rho_{P})*p'\circ \jnedafinal_{R}^{P*R}(\rho_R))*p''\circ \jnedafinal_{Q}^{P*R*Q}(\rho_Q).
 \end{equation}
 To show that $\Upsilon(\mathbf{e}_{P}^{P*R*Q})(p)=\bigl(\Upsilon(\mathbf{e}_{P*R}^{P*R*Q})\circ \Upsilon(\mathbf{e}_{P}^{P*R})\bigl)(p)$, we need to show that (\ref{proofeq1})=(\ref{proofeq2}). To show the equality of the second paths in each equation, we prove that
    \begin{equation}\label{eq: in the proof}
        p''\circ \jnedafinal_{R*Q}^{P*R*Q} \circ \jneda_{R}^{R*Q}=p'\circ \jnedafinal_{R}^{P*R}
    \end{equation}
    Using the expression of initial and final inclusions (Lem.\@ \ref{lemma: pomset pushout}), an elementary calculation shows that $\jnedafinal_{R*Q}^{P*R*Q} \circ \jneda_{R}^{R*Q}=\jneda_{P*R}^{P*R*Q}\circ \jnedafinal_{R}^{P*R}$, meaning that the square in the following diagram commutes. Thus, $ p''\circ \jnedafinal_{R*Q}^{P*R*Q} \circ \jneda_{R}^{R*Q}= p''\circ \jneda_{P*R}^{P*R*Q}\circ \jnedafinal_{R}^{P*R}$ and (\ref{eq: in the proof}) follows since $p''\circ \jneda_{P*R}^{P*R*Q} =p'$
       \[ \xymatrix{
    & \pobj{R} \ar[r]^{\jneda} \ar[d]^{\jnedafinal}&  {\pobj{R*Q}} \ar[d]^{\jnedafinal}  \\ \pobj{P} \ar[dr]_{p} \ar[r]^{\jneda} &  {\pobj{P*R}} \ar[r]^{\jneda} \ar[d]_{p'} &  \pobj{P*R*Q}  \ar[dl]^{p''}  \\   & X & } \]
    what left is to show the third paths in each equation are equal, i.e, 
    \[ p''\circ \jnedafinal_{R*Q}^{P*R*Q}\circ \jnedafinal_{Q}^{R*Q}(\rho_Q)=p''\circ \jnedafinal_{Q}^{P*R*Q}(\rho_Q) \]
    which is clear by the following diagram
    \[ \xymatrix{
     &  {\pobj{Q}} \ar[r]^{\jnedafinal} \ar[d]_{\jnedafinal} &  \pobj{R*Q}  \ar[dl]^{\jnedafinal}  \\   & \pobj{P*R*Q}& } \]
    \end{proof}
    \begin{proof}[Proof of Thm. \@ \ref{th: last theorem}]
    $\Rightarrow$ Assume that there is a (strong) track-bisimulation $K$ between $\mathcal{X}_1$ and $\mathcal{X}_2$. We show that the relation between tracks with domain $\square^P$ given by $\R=\{\bigl(p_1(\rho_P),p_2(\rho_P)\bigl) \mid (p_1,p_2)\in K\}$, where $\rho_P$ is the characteristic path of $P$, is a (strong) path-bisimulation. By Yoneda lemma, there is a one-to-one correspondence between the initial cells and the initial tracks, thus $R$ satisfies \ref{def: U bisimulation}.\ref{en: U bisim.e1}. Since $p_i$ are precubical maps, \ref{def: U bisimulation}.\ref{en: U bisim.e3} is satisfied. Note that $\R=\{\bigl( \Upsilon(p_1),\Upsilon(p_1)\bigl) \mid (p_1,p_2)\in K\}$, where $\Upsilon$ is the isomorphism of Lemma \ref{lemma: Upsilon is a functor}. 
    The property \ref{def: U bisimulation}.\ref{en: U bisim.e4}
    holds because $\Upsilon$ is an isomophism between the categories $\mathbb{P}_{X_i}$ and $T_{X_i}$ (Th.\@ \ref{th: path and tracks isomorphism}) having extensions as morphisms and by the the property \ref{def: Track-bisimulation}.\ref{eq2: track bis}.

    $\Leftarrow$ For any path $\alpha \in \mathrm{P}_{X}$, we denote by $p_{\alpha}: \pobj{\ev(\alpha)} \to X$ the track of Prop.\@ \ref{prop: generalization of Yoneda}. We assume that there is a strong Path-bisimulation $\R$ between HDAs $\mathcal{X}_1$ and $\mathcal{X}_2$. We show that the relation between tracks $K=\{(p_{\alpha_1},p_{\alpha_2})\mid (\alpha_1,\alpha_2) \in K \}$ is a strong $\TrO^0_S$-bisimulation. First, $K$ satisfies \ref{def: Track-bisimulation}.\ref{eq1: track bis} because there is one-to-one correspondence between initial tracks and initial cells (Yoneda embedding). To show \ref{def: Track-bisimulation}.\ref{eq2: track bis}, it is enough to notice that $K=\{\bigl(\Psi(\alpha_1),\Psi(\alpha_2)\bigl)\mid (\alpha_1,\alpha_2) \in K \}$ and to use the property \ref{def: U bisimulation}.\ref{en: U bisim.e4}.
\end{proof}

\end{document}